\documentclass[copyright,creativecommons]{eptcs}
\providecommand{\event}{LSFA 2012} 
\usepackage{ifthen}
\newboolean{talk}
\setboolean{talk}{false}
\newboolean{paper}
\setboolean{paper}{false}
\newboolean{french}
\setboolean{french}{false}
\newboolean{IEEEstyle}
\setboolean{IEEEstyle}{false}
\newboolean{lipicsstyle}
\setboolean{lipicsstyle}{false}
\newboolean{eptcsstyle}
\setboolean{eptcsstyle}{false}

\setboolean{eptcsstyle}{true}
\ifthenelse{\boolean{talk}}{}{\usepackage[utf8]{inputenc}}
\ifthenelse{\boolean{french}}{\usepackage[french]{babel}}{
  \ifthenelse{\boolean{lipicsstyle}}{}{\usepackage[english]{babel}}}
\usepackage{amssymb}
\usepackage{amsmath}
\usepackage{graphicx}
\usepackage{bussproofs}
\usepackage{fancybox}
\usepackage{cmll}
\usepackage{stmaryrd}
\usepackage{multirow}
\ifthenelse{\boolean{talk}}{\usepackage{color}}{
  \ifthenelse{\boolean{lipicsstyle}}{\usepackage{color}}{\usepackage[usenames]{color}}}

\ifthenelse{\boolean{IEEEstyle}}{
\usepackage{amsthm}
\theoremstyle{plain}
    \newtheorem{theorem}{Theorem}[section]
    \newtheorem{lemma}[theorem]{Lemma}
    \newtheorem{corollary}[theorem]{Corollary}
    
\theoremstyle{definition}
    \newtheorem{definition}[theorem]{Definition}
\theoremstyle{remark}
    \newtheorem{remark}[theorem]{Remark}
}{}

\ifthenelse{\boolean{eptcsstyle}}{
}{}

\newcommand{\tm}{t}
\newcommand{\tmtwo}{s}
\newcommand{\tmthree}{u}

\newcommand{\var}{x}

\newcommand{\Rew}[1]{\rightarrow_{#1}}

\renewcommand{\to}{\Rew{}}

\newcommand{\eqw}[1]{\equiv_{#1}}
\newcommand{\CS}{{CS}}
\newcommand{\eqvo}{\eqw{\vosym}}

\newcommand{\vosym}{vo}

\newcommand{\msym}{{\tt m}}
\newcommand{\esym}{{\tt e}}

\newcommand{\val}{v}

\newcommand{\nbvctxtwo}[1]{\nbvctxtwo{#1}}

\newcommand{\defeq}{:=}

\newcommand{\isub}[2]{\{#1/#2\}}
\newcommand{\esub}[2]{[#1/#2]}
\renewcommand{\L}{{\tt L}}

\newcommand{\llbrace}{\{ \kern -0.27em \vert}
\newcommand{\rrbrace}{\vert \kern -0.27em \}}

\renewcommand{\l}{\lambda}
\newcommand{\ie}{{\em i.e.}}

\newcommand{\ih}{\textit{i.h.}}
\newcommand{\fv}[1]{{\tt fv}(#1)}

\newcommand{\deff}[1]{\textbf{#1}}

\newcommand{\ignore}[1]{}
\newcommand{\sep}{\hspace*{0.5cm}}
\newcommand{\ms}{\medskip}
\newcommand{\mathintitle}[2]{\texorpdfstring{#1}{#2}}

\newcommand{\myinput}[1]{\ifthenelse{\boolean{withimages}}{\input{#1}}{}}
\newcommand{\spc}{@{\hspace{.5cm}}}

\newcommand{\linlogic}{linear logic}

\newcommand{\der}{{\mathsf{d}}}

\newcommand{\weak}{{\mathsf{w}}}

\newcommand{\set}[1]{\{#1\}}

\newcommand{\nat}{\mathbb{N}}

\newcommand{\pns}{proof nets}

\newcommand{\zeronet}[1]{#1^0}

\newcommand{\bbox}[1]{box(#1)}

\newcommand{\preeqw}[1]{\sim_{#1}}

\newcommand{\kingv}[1]{king(#1)}

\newcommand{\tom}{\Rew{\msym}}
\newcommand{\toe}{\Rew{\esym}}

\newcommand{\lv}{\l_{\beta v}}

\newcommand{\lvsub}{\l_{vsub}}
\newcommand{\lvker}{\l_{vker}}

\newcommand{\lcbv}{\l_{CBV}}

\newcommand{\cbn}[1]{#1^{\tt n}}
\newcommand{\cbnp}[1]{\cbn{(#1)}}
\newcommand{\cbv}[1]{#1^{\tt v}}
\newcommand{\cbvp}[1]{\cbv{(#1)}}

\newcommand{\kt}[1]{#1^{\tt k}}

\newcommand{\lamtonets}[1]{\underline{#1}}
\newcommand{\lamtonetsvar}[2]{\underline{#1}_{#2}}

\newcommand{\lt}[1]{\lamtonets{#1}}

\usepackage{pict2e}
\usepackage{tikz}
 \usetikzlibrary{calc}
 \usetikzlibrary{backgrounds}
\usetikzlibrary{decorations.shapes}
 \usetikzlibrary{decorations.text}
\usetikzlibrary{decorations.pathmorphing}
\usetikzlibrary{decorations.markings}
 \usetikzlibrary{matrix}
\usetikzlibrary{decorations.pathreplacing}
 \usepackage[nofancy]{tikz-inet}

\tikzset{%
  ocenter/.style={baseline={([yshift=-.5ex, xshift=-.5ex]current bounding box)}},
  dummynode/.style={inner sep= 0.5pt, anchor= center},
acout/.style={inner sep= 0.6pt, anchor= center,circle,draw=black,fill=white, minimum size = 1pt},
  oldout/.style={inner sep= 0.6pt, anchor= center,circle,draw=black,fill=black!50, minimum size = 1pt},
  ln/.style={->, shorten >=2pt, shorten <=1pt, draw=black!50, line width=0.2ex},
  nl/.style={->, shorten <=2pt, shorten >=1pt, draw=black!50, line width=0.2ex},
  linknode/.style={->, shorten >=2pt, shorten <=1pt, densely dotted,draw=black!50, line width=0.2ex},
  nodelink/.style={->, shorten <=2pt, shorten >=1pt, densely dotted,draw=black!50, line width=0.2ex},
  negln/.style={->, densely dotted, shorten >=1pt, shorten <=1pt, draw=red, line width=0.18ex},
  negnl/.style={->, densely dotted, shorten <=1pt, shorten >=1pt, draw=red, line width=0.18ex},
  posln/.style={->, shorten >=1pt, shorten <=1pt, draw=blue, line width=0.18ex},
  posnl/.style={->, shorten <=1pt, shorten >=1pt, draw=blue, line width=0.18ex},
  nopol/.style={->, shorten <=1pt, shorten >=1pt, draw=gray, line width=0.18ex},
  nopolrev/.style={<-, shorten <=1pt, shorten >=1pt, draw=gray, line width=0.18ex},
  pure/.style={->, shorten <=1pt, shorten >=1pt, draw=brown, line width=0.18ex},
  exppure/.style={->, shorten <=1pt, shorten >=1pt, draw=cyan, line width=0.18ex},
  enaryedge/.style={preaction={decorate},decoration={markings,mark=at position .5 with {\draw [shorten >=0pt, shorten <=0pt,draw=cyan,-](0pt,3pt) -- (0pt,-3pt);}}},
mnaryedge/.style={preaction={decorate},decoration={markings,mark=at position .5 with {\draw [shorten >=0pt, shorten <=0pt,draw=brown,-](0pt,3pt) -- (0pt,-3pt);}}},
  neutral/.style={->, shorten <=1pt, shorten >=1pt, draw=black, line width=0.18ex, dashed},  
  und/.style={ shorten <=1pt, shorten >=1pt, draw=gray, line width=0.18ex},
  princ/.style={postaction={decorate},decoration={markings,mark=at position .5 with {\node [inner sep= 0.6pt, anchor= center,circle, draw=brown, minimum size = 0.8pt, solid, line width=0.10ex]{};}}},
  mprinc/.style={postaction={decorate},decoration={markings,mark=at position .4 with {\node [inner sep= 0.7pt, anchor= center,circle,draw=brown, fill=brown, minimum size = 0.9pt, solid, line width=0.10ex]{};}}},
  eprinc/.style={postaction={decorate},decoration={markings,mark=at position .4 with {\node [inner sep= 0.7pt, anchor= center,circle,draw=cyan, fill=cyan, minimum size = 0.9pt, solid, line width=0.10ex]{};}}},
  genred/.style={preaction={decorate},decoration={markings,mark=at position .5 with {\draw [shorten >=0pt, shorten <=0pt,draw=red,-](0pt,3pt) -- (0pt,-3pt);}}},
  gengray/.style={preaction={decorate},decoration={markings,mark=at position .5 with {\draw [shorten >=0pt, shorten <=0pt,draw=gray,-](0pt,3pt) -- (0pt,-3pt);}}},
  genblue/.style={preaction={decorate},decoration={markings,mark=at position .5 with {\draw [shorten >=0pt, shorten <=0pt,draw=blue,-](0pt,3pt) -- (0pt,-3pt);}}},
  genbluet/.style={preaction={decorate},decoration={markings,mark=at position .75 with {\draw [shorten >=0pt, shorten <=0pt,draw=blue,-](0pt,3pt) -- (0pt,-3pt);}}},
  genblues/.style={preaction={decorate},decoration={markings,mark=at position .25 with {\draw [shorten >=0pt, shorten <=0pt,draw=blue,-](0pt,3pt) -- (0pt,-3pt);}}},
  genexppure/.style={preaction={decorate},decoration={markings,mark=at position .5 with {\draw [shorten >=0pt, shorten <=0pt,draw=cyan,-](0pt,3pt) -- (0pt,-3pt);}}},
  negnlgen/.style={preaction={decorate},decoration={markings,mark=at position .5 with {\draw [shorten >=0pt, shorten <=0pt,draw=red,-](0pt,3pt) -- (0pt,-3pt);}},->, shorten <=1pt, shorten >=1pt, densely dotted, draw=red, line width=0.20ex},
  neglngen/.style={preaction={decorate},decoration={markings,mark=at position .5 with {\draw [shorten >=0pt, shorten <=0pt,draw=red,-](0pt,3pt) -- (0pt,-3pt);}},->, shorten <=1pt, shorten >=1pt, densely dotted, draw=red, line width=0.20ex},
  posnlgen/.style={preaction={decorate},decoration={markings,mark=at position .5 with {\draw [shorten >=0pt, shorten <=0pt,draw=red,-](0pt,3pt) -- (0pt,-3pt);}},->, shorten <=1pt, shorten >=1pt, draw=blue, line width=0.20ex},
  poslngen/.style={preaction={decorate},decoration={markings,mark=at position .5 with {\draw [shorten >=0pt, shorten <=0pt,draw=red,-](0pt,3pt) -- (0pt,-3pt);}},->, shorten <=1pt, shorten >=1pt, draw=blue, line width=0.20ex},
  linknodeNor/.style={ shorten >=2pt, shorten <=1pt, draw=black!50, line width=0.10ex},
  nodelinkNor/.style={ shorten <=2pt, shorten >=1pt, draw=black!50, line width=0.10ex},
  linknodeMark/.style={ shorten >=2pt, shorten <=1pt, draw=black!50, line width=0.30ex},
  nodelinkMark/.style={ shorten <=2pt, shorten >=1pt, draw=black!50, line width=0.30ex},
  nodelinkgen/.style={densely dotted, preaction={decorate},decoration={markings,mark=at position .5 with {\draw [shorten >=0pt, shorten <=0pt,-](0pt,3pt) -- (0pt,-3pt);}},->, shorten <=2pt, shorten >=1pt, draw=black!50, line width=0.20ex},
  posnlgen/.style={preaction={decorate},decoration={markings,mark=at position .5 with {\draw [shorten >=0pt, shorten <=0pt,-](0pt,3pt) -- (0pt,-3pt);}},->, shorten <=2pt, shorten >=1pt, draw=black!50, line width=0.18ex},
    poslngen/.style={preaction={decorate},decoration={markings,mark=at position .5 with {\draw [shorten >=0pt, shorten <=0pt,-](0pt,3pt) -- (0pt,-3pt);}},->, shorten >=2pt, shorten <=1pt, draw=black!50, line width=0.18ex},
  nodelinkgenMark/.style={postaction={decorate},decoration={markings,mark=at position .5 with {\draw [shorten >=0pt, shorten <=0pt,-](0pt,3pt) -- (0pt,-3pt);}}, shorten <=2pt, shorten >=1pt, draw=black!50, line width=0.30ex, },  
  nospace/.style={inner sep= 0pt, anchor=center},
  etic/.style={inner sep= 0.5pt, fill=white, anchor= center},
  punto/.style={fill=gray, circle, inner sep= 0.7pt, anchor= center},
  port/.style={inner sep= 1pt, anchor= center},
  emp/.style={inner sep= 4pt, fill=white, anchor= center, circle, draw=black, densely dotted},
  corner/.style={inner sep= 0pt, anchor= center},
  eport/.style={inner sep= 0.8pt, anchor= center,circle,draw=cyan,fill=white, minimum size = 0.8pt},
  mport/.style={inner sep= 0.8pt, anchor= center, circle,draw=white,fill=brown, minimum size = 0.8pt, solid, 
line width=0.10ex},
  posport/.style={inner sep= 0.8pt, anchor= center,circle,draw=blue,fill=white, minimum size = 0.8pt},
  negport/.style={inner sep= 0.8pt, anchor= center, circle,draw=red,fill=white, minimum size = 0.8pt, solid, line width=0.10ex},  
  port/.style={inner sep= 0.5pt, anchor= center, circle,draw=black,fill=black, minimum size = 0.5pt},
  biport/.style={inner sep= 0.5pt, anchor= center, circle,draw=black,fill=white, minimum size = 2pt},
  negat/.style={fill=gray, circle, inner sep= 0.7pt, anchor= center},
  net/.style={draw=gray,inner sep=4pt,thick,ellipse, anchor=center},
  smallnet/.style={draw=gray,inner sep=3pt,thick,ellipse },
  bignet/.style={draw=gray,inner sep=12pt,thick,ellipse},
  box/.style={rectangle, draw= gray,rounded corners=2pt, inner sep= 1pt, line width=0.30ex},
  boxangle/.style={inner sep= 0.5pt},
  noboxline/.style={draw= white,rounded corners, line width=0ex},
  jboxline/.style={draw= gray,rounded corners, line width=0.20ex, overlay},
  exboxline/.style={draw= gray,line width=0.20ex, overlay},
  inductiveTr/.style={draw=black!50, minimum size=0.9cm},
  inductiveTrs/.style={draw=black!50, minimum size=0.6cm},
  vcenter/.style={baseline={([yshift=-.5ex]current bounding box)}},  
  jump/.style={ line width=0.2ex, shorten >=3pt, shorten <=2pt, draw=black!50, line width=0.30ex, ->},
    jumpblue/.style={postaction={decorate,decoration={markings,mark=at position .4 with {\node[etic,fill=white]{{$\jop$}};}}},  line width=0.2ex, shorten >=3pt, shorten <=2pt,  line width=0.30ex, ->, draw=blue},
    jumpnopol/.style={postaction={decorate,decoration={markings,mark=at position .4 with {\node[etic,fill=white]{{$\jop$}};}}}, line width=0.2ex, shorten >=3pt, shorten <=2pt, draw=gray, densely dotted, ->},
  jumpsub/.style={postaction={decorate,decoration={markings,mark=at position .4 with {\node[etic,fill=white]{{$\jop$}};}}}, line width=0.2ex, shorten >=3pt, shorten <=2pt, draw=blue, ->},
  jumpshift/.style={postaction={decorate,decoration={markings,mark=at position .4 with {\node[etic,fill=white]{{\tiny s}};}}}, line width=0.2ex, shorten >=3pt, shorten <=2pt, draw=black!50, line width=0.30ex, ->},
  jumpcut/.style={postaction={decorate,decoration={markings,mark=at position .5 with {\node[etic,fill=white]{{\tiny c}};}}}, line width=0.2ex, shorten >=3pt, shorten <=2pt, draw=blue, line width=0.18ex, ->},
  jumpbord/.style={postaction={decorate,decoration={markings,mark=at position .4 with {\node[etic,fill=white]{{\tiny b}};}}}, line width=0.2ex, shorten >=3pt, shorten <=2pt, draw=black!50, line width=0.30ex, ->},
  jumpzero/.style={postaction={decorate,decoration={markings,mark=at position .4 with {\node[etic,fill=white]{{\tiny 0}};}}}, line width=0.2ex, shorten >=3pt, shorten <=2pt, draw=black!50, line width=0.30ex, ->},
   jumpgen/.style={preaction={decorate},decoration={markings,mark=at position .7 with {\draw [shorten >=0pt, shorten <=0pt,-](0pt,3pt) -- (0pt,-3pt);}}, line width=0.2ex, shorten >=2.5pt, shorten <=1pt, draw=black!50, line width=0.30ex, ->},
  jumpcutgen/.style={preaction={decorate},decoration={markings,mark=at position .5 with {\draw [shorten >=0pt, shorten <=0pt,-](0pt,3pt) -- (0pt,-3pt);}},postaction={decorate,decoration={markings,mark=at position .4 with {\node[etic,fill=white]{{\tiny c}};}}}, line width=0.2ex, shorten >=3pt, shorten <=2pt, draw=blue, line width=0.18ex, ->},
  jumpbordgen/.style={preaction={decorate},decoration={markings,mark=at position .7 with {\draw [shorten >=0pt, shorten <=0pt,-](0pt,3pt) -- (0pt,-3pt);}},postaction={decorate,decoration={markings,mark=at position .4 with {\node[etic,fill=white]{{\tiny b}};}}}, line width=0.2ex, shorten >=3pt, shorten <=2pt, draw=black!50, line width=0.30ex, ->},
   jumpzerogen/.style={preaction={decorate},decoration={markings,mark=at position .7 with {\draw [shorten >=0pt, shorten <=0pt,-](0pt,3pt) -- (0pt,-3pt);}},postaction={decorate,decoration={markings,mark=at position .4 with {\node[etic,fill=white]{{\tiny 0}};}}}, line width=0.2ex, shorten >=3pt, shorten <=2pt, draw=black!50, line width=0.30ex, ->},
  arrowed/.style={preaction={decorate},decoration={markings,mark=at position .55 with {\arrow{>}}}, thin},
  jumpLink/.style={etic,fill=white},
  revArrowed/.style={preaction={decorate},decoration={markings,mark=at position .55 with {\arrowreversed{>}}}, thin},
  revBiarrowed/.style={preaction={decorate},decoration={markings,mark=at position .45 with {\arrowreversed{<<}}}, thin},
  cell/.style={minimum size=0.6cm},
  every label/.style={label distance = 3pt, font=\tiny, inner sep= 1pt},  
  every node/.style={font=\scriptsize }
}

\newcommand{\stalt}{26pt}
\newcommand{\stlar}{20pt}
\newcommand{\hstalt}{\stalt/2}
\newcommand{\hstlar}{\stlar/2}

\newcommand{\lparv}[4]
{
\lbinsym{#1}{#2}{#3}{#4}{$\parr$}
\draw[exppure, in=210, out=90] (#2) to (#4);
\draw[exppure, out=-30, in=90] (#4) to (#3);
\draw[mprinc,pure] (#1) to (#4);
}

\newcommand{\lparvangle}[5]{
\lbinfixr{#1}{#2}{#3}{#4}{$\parr$}
\draw[exppure, #5] (#2) to (#4);
\draw[exppure, out=-30, in=90] (#4) to (#3);
\draw[mprinc, pure] (#1) to (#4);
}

\newcommand{\lparvanglefake}[5]{
\lbinfixr{#1}{#2}{#3}{#4}{}
}

\newcommand{\ltensv}[4]
{
\lbinsym{#1}{#2}{#3}{#4}{$\otimes$}
\draw[mprinc,pure, out=210, in=90] (#4) to (#2);
\draw[exppure, out=-30, in=90] (#4) to (#3);
\draw[exppure] (#1) to (#4);

}
\newcommand{\lbangv}[3]
{
\node at ($(#2.center) ! .5 ! (#1.center)$) [above= .5pt,fill=white, etic] (#3){\scriptsize $!$};

\draw[eprinc, exppure] (#2) to (#3);
\draw[pure] (#3) to (#1);
}

\newcommand{\lbangvfake}[3]
{
\node at ($(#2.center) ! .5 ! (#1.center)$) [above= .5pt,fill=white, etic] (#3){\scriptsize $!$};

}

\newcommand{\lderv}[3]
{
\node at ($(#2.center) ! .5 ! (#1.center)$) [above= .5pt,fill=white, etic] (#3){\scriptsize $\der$};

\draw[pure] (#2) to (#3);
\draw[eprinc, exppure] (#3) to (#1);
}

\newcommand{\ldervangle}[4]
{
\node at ($(#2.center) ! .5 ! (#2 |- #1)$) [above= .5pt,fill=white, etic] (#3){\scriptsize $\der$};

\draw[pure] (#2) to (#3);
\draw[eprinc, exppure, #4] (#3) to (#1);
}

\newcommand{\lweakv}[2]
{
\node at (#1.center) [above= \hstalt,fill=white, etic] (#2){\scriptsize$\weak$};

\draw[eprinc, exppure] (#2) to (#1);

}
\newcommand{\lcolbox}[4]{
\node at ($(#2.center) ! .5 ! (#3.center) ! .5 ! (#1.center)$) [below= 1pt, etic] (#4){$\square$};
\draw[pure, in=90, out=-90] (#1) to (#4);
\draw[exppure, in=90, out=210] (#4) to (#2);
\draw[exppure, in=90, out=-30] (#4) to (#3);

}
\newcommand{\lbinsym}[5]{
\node at ($(#2.center) ! .5 ! (#3.center) ! .5 ! (#1.center)$) [etic] (#4){#5};
}

\newcommand{\lbinfixr}[5]{
\node at ($(#1.center) ! .5 ! (#3 -| #1)$) [below= 1pt, etic] (#4){#5};
}

\newcommand{\stboxlw}{4pt}

\newcommand{\sepbox}{5pt}
\newcommand{\sepboxshort}{1pt}

\newcommand{\boxinit}{
\pgfdeclarelayer{background layer} 
\pgfdeclarelayer{foreground layer} 
\pgfsetlayers{background layer,main,foreground layer} 
}

\newcommand{\boxlinev}[2]{
\draw[#2](#1so.center) -- (#1se.center) -- (#1ne.center) -- (#1.center) -- (#1no.center) -- (#1so.center); 
}

\newcommand{\aboxvnodes}[5]{
\boxvnodes{#1}{#2}{#4}{#5}
\boxlinev{#1}{#3line}
}

\newcommand{\boxvnodes}[4]{
\node at (#1.center)[left = #3,nospace](#1no){};
\node at (#1.center)[right = #4,nospace](#1ne){};
\node at (#1no|-#2)[nospace](#1so){};
\node at (#1ne|-#2)[nospace](#1se){};
}

\newcommand{\jaux}[2]{
\node at (#1se.center)[left = \stboxlw, negport](#1auxk){};
\node at (#1auxk.center)[left= ((#2)/2),negport](#1aux1){};
}

\newcommand{\exaux}[2]{
\node at (#1se.center)[left = (0.15*#2), etic](#1auxk){\tiny $\aux$};
\node at (#1se.center)[left=0.6*#2, etic](#1aux1){\tiny $\aux$};
}

\newcommand{\boxaux}[3]{
\ifthenelse{\equal{#3}{jbox}}{\jaux{#1}{#2}}{\exaux{#1}{#2}};
}

\newcommand{\gdots}[3]{
\node at ($(#1)!.5!(#2)$) [#3](#1#2dots){\tiny $\ldots$};
}

\newcommand{\med}[2]{
($(#1)!.5!(#2)$)
}

\title{Proof nets and the call-by-value $\l$-calculus}
\author{Beniamino Accattoli
\institute{INRIA \& \'Ecole Polytechnique (LIX), France}}
\def\titlerunning{Proof nets and the call-by-value $\l$-calculus}
\def\authorrunning{B. Accattoli}

\newboolean{withimages}
\setboolean{withimages}{true}

\begin{document}
\maketitle

\begin{abstract}
This paper gives a detailed account of the relationship between (a variant of) the call-by-value lambda calculus and linear logic proof nets. The presentation is carefully tuned in order to realize a strong bisimulation between the two systems: every single rewriting step on the calculus maps to a single step on the nets, and viceversa. In this way, we obtain an algebraic reformulation of proof nets. Moreover, we provide a simple correctness criterion for our proof nets, which employ boxes in an unusual way.
\end{abstract}

\section{Introduction}
A key feature of \linlogic\ (LL) is that it is a refinement of intuitionistic logic, \ie\ of $\l$-calculus. In particular, \emph{one} $\beta$-reduction step in the $\l$-calculus corresponds to the sequence of \emph{two}  cut-elimination steps in \linlogic, steps which are of a very different nature: the first is multiplicative and the second is exponential. The Curry-Howard interpretation of this fact is that $\l$-calculus can be refined adding a constructor $t[x/u]$ for \emph{explicit substitutions}, and decomposing a $\beta$-step $(\l x.t)u\Rew{\beta} t\isub{x}{u}$ into the sequence $(\l x.t)u\tom t[x/u] \toe t\isub{x}{u}$. 

Another insight due to \linlogic\ is that proofs can be represented graphically---by the so-called \pns---and the reformulation of cut-elimination on \pns takes a quite different flavour with respect to cut-elimination in sequent calculus. The parallel nature of the graphical objects makes the commutative cut-elimination steps, which are the annoying burden of every proof of cut-admissibility, (mostly) disappear. 

These two features of LL have influenced the theory of explicit substitutions in various ways \cite{DBLP:journals/iandc/KesnerL07,DBLP:journals/mscs/CosmoKP03}, culminating in the design of \emph{the structural $\l$-calculus} \cite{DBLP:conf/csl/AccattoliK10}, a calculus isomorphic (more precisely \emph{strongly bisimilar}) to its representation in LL \pns\ \cite{DBLP:conf/csl/AccattoliG09,phdaccattoli}. Such a calculus can be seen as an algebraic reformulation of \pns\ for $\l$-calculus \cite{Danos:Thesis:90,Reg:Thesis:92}, and turned out to be simpler and more useful than previous calculi with explicit substitutions.

Girard's seminal paper on \linlogic\ \cite{DBLP:journals/tcs/Girard87} presents two translations of $\l$-calculus into LL.  The first one follows the typed scheme $\cbnp{A\Rightarrow B}=!\cbn{A}\multimap \cbn{B}$, and it is the one to which the previous paragraphs refer to. It represents the ordinary---or call-by-name (CBN)---$\l$-calculus. The second one, identified by $\cbvp{A\Rightarrow B}=!(\cbv{A}\multimap \cbv{B})$, was qualified as \textit{boring} by Girard and received  little attention in the literature \cite{DBLP:journals/tcs/MaraistOTW99,DBLP:journals/mscs/PravatoRR99,DBLP:journals/tcs/DanosR99,DBLP:conf/gg/FernandezM02,DBLP:journals/corr/abs-1003-5515,DBLP:conf/ifl/Mackie05}. Usually, it is said to represent Plotkin's call-by-value (CBV) $\lv$-calculus \cite{DBLP:journals/tcs/Plotkin75}. These two representations concern typed terms only, but it is well-known that they can be extended to represent the whole untyped calculi by considering linear recursive types ($o=!o\multimap o$ for call-by-name 
and and $o=!(o\multimap o)$ for call-by-value).

Surprisingly, the extension of the CBV translation to the untyped calculus $\lv$-calculus introduces a violent unexpected behavior: some normal terms in $\lv$ map to (recursively typed) \pns\ without normal form (see \cite{AccLinearity} for concrete examples and extensive discussions). This fact is the evidence that there is something inherently wrong in the CBV translation.

In this paper we show how to refine the three actors of the play (the CBV $\l$-calculus, the translation and the \pns\ presentation) in order to get a perfect match between terms and \pns. Technically, we show that the new translation is a strong bisimulation\footnote{A strong bisimulation between two rewriting systems $S$ and $R$ is a relation $\equiv$ between $S$ and $R$ s.t. whenever $s\equiv r$ then for every step from $s \Rew{S} s'$ there is a step $r\Rew{R} r'$ s.t. $s'\equiv r'$, \emph{and viceversa} (for $s,s'\in S$ and $r,r'\in R$).}, and since strong bisimulations preserve reductions length (in both directions), the normalization mismatch vanishes.

Interestingly, to obtain a strong bisimulation we have to make some radical changes to both the calculus and the presentation of \pns. The calculus, that we call the \emph{value substitution kernel} $\lvker$ \cite{AccLinearity}, is a subcalculus of the \emph{value substitution calculus} $\lvsub$ studied in  \cite{DBLP:conf/flops/AccattoliP12}, which is a CBV $\l$-calculus with explicit substitutions. Such a kernel is as expressive as the full calculus, and can be thought as a sort of CPS representation of $\lvsub$. 

Here, however, we mostly take the calculus for granted (see \cite{AccLinearity} for more details) and rather focus on \pns. Our two contributions are:

\begin{enumerate}
 \item \deff{Graphical syntax and algebraic formalism}: it is far from easy to realize a strong bisimulation between terms and nets, as it is necessary to take care of many delicate details about weakenings, contractions, representation of variables, administrative reduction steps, and so on. The search for a strong bisimulation may seem a useless obsession, but it is not. Operational properties as confluence and termination then transfer immediately from graphs to terms, and viceversa. More generally, such a strong relationship turns the calculus into an algebraic language for \pns, providing an handy tool to reason by structural induction over \pns.

\item \deff{Correctness criterion}: we provide a characterization of the \pns\ representing $\lvker$ based on graph-theoretical principles and which does not refer to $\lvker$, that is, we present a \emph{correctness criterion}. Surprisingly, the known criteria for the representation of the call-by-name $\l$-calculus (with explicit substitutions) fail to characterize the fragment encoding the call-by-value $\l$-calculus. Here we present a simple and non-standard solution to this problem. We hack the usual presentation of \pns\ so that Laurent's criterion for polarized nets \cite{DBLP:conf/tlca/Laurent99,DBLP:journals/tcs/Laurent03,phdlaurent}---the simplest known correctness criterion---captures the fragment we are interested in. The hacking of the syntax consists in using boxes for $\parr$-links rather than for $!$-links. An interesting point is that the fragment we deal with is not polarized in Laurent's sense, despite it is polarized in the Lamarche/intuitionistic sense.
\end{enumerate}

The use of boxes for $\parr$-links may look terribly ad-hoc. Section \ref{s:par} tries to argue that it is not. Moreover, Section \ref{s:history} presents an account of the technical points concerning the representations of terms with \pns, and how they have been treated in the literature.
\section{Terms}
In this section we introduce the calculus which will be related to proof nets, called \emph{the value substitution kernel} $\lvker$ \cite{AccLinearity}. Its syntax is:
\begin{center}
$\begin{array}{ccc@{\sep\sep\sep\sep}ccc}
 t,s,u,r&::=& x \mid \l x. t\mid \val s\mid t[x/u] &\val&::=& x\mid \l x. t
\end{array}
$\end{center}
where $t[x/u]$ is an \emph{explicit substitution} and values are noted $\val$. Note that the left subterm of an application can only be a value. The rules of $\lvker$ are:
\begin{center}
$\begin{array}{ccc@{\sep\sep\sep\sep}ccc}
 (\l x. t) u &\mapsto_{\msym} &t[x/u]&
t[x/\val\L]&\mapsto_{\esym} & t\isub{x}{\val}\L
\end{array}$\end{center}
where $\L$ is a possibly empty list of explicit substitutions $[x_1/u_1]\ldots [x_k/u_k]$ (and the fact that in the lhs of $\mapsto_{\esym}$ $\L$ appears inside $[\ ]$ while in the rhs it appears outside $\{\ \}$ is not a typo). The calculus is confluent \cite{AccLinearity}.

The peculiarity of the value substitution kernel is that iterated applications as $(tu)s$ are not part of the language. The idea is that they are rather represented as $(x s)[x/t u]$ with $x$ fresh. The calculus containing iterated applications is called \emph{the value substitution calculus} $\lvsub$, and it has been studied in \cite{DBLP:conf/flops/AccattoliP12,AccLinearity}. In \cite{AccLinearity} it is shown that $\lvsub$ can be represented inside $\lvker$ (mapping iterated applications $(tu)s$ to $(x s)[x/t u]$, as before) and that a term $t$ and its representation $\kt{t}$ are equivalent from the point of view of termination (formally $t$ is strongly (resp. weakly) normalizing iff $\kt{t}$ is, and the same is true with respect to weak---\ie\ not under lambda---reduction). If one is interested in observing termination (as it is usually the case) than $\lvsub$ and $\lvker$ are observationally equivalent (via $\kt{\cdot}$). As pointed out to us by Frank Pfenning, the map $\kt{\cdot}$ is reminiscent of the 
notion of \emph{$A$-reduction} in the theory of CPS-translations \cite{DBLP:conf/pldi/FlanaganSDF93,DBLP:journals/lisp/SabryF93}. The idea is then that $\lvker$ (and thus proof nets) is essentially the language of $A$-normal forms associated to $\lvsub$. However, the study of the precise relationship with $A$-normal forms is left to future work.

The calculus $\lvsub$ has been related to Herbelin and Zimmermann's $\lcbv$ \cite{DBLP:conf/tlca/HerbelinZ09} in \cite{DBLP:conf/flops/AccattoliP12}. In turn, $\lcbv$ is related to Plotkin's $\lv$ in \cite{DBLP:conf/tlca/HerbelinZ09}, where it is shown that the equational theory of $\lv$ is contained in the theory of $\lcbv$.\ms

The rest of the paper shows that $\lvker$ can be seen as an algebraic language for the proof nets used to interpret the call-by-value $\l$-calculus.
\section{Proof nets: definition}
\label{s:proofnets}
\emph{Introduction}. Our presentation of proof nets is non-standard in at least four points (we suggest to have a quick look to Figure \ref{fig:trans}):
\begin{enumerate}
\item \textbf{Hypergraphs}: we use hypergraphs (for which formulas are nodes and links---\ie\  logical rules---are hyperedges) rather than the usual graphs with pending edges (for which formulas are edges and links are nodes). We prefer hypergraphs because in this way contraction can be represented in a better way (providing commutativity, associativity, and permutation with box borders \emph{for free}) and at the same time we can represent cut and axiom links implicitly (similarly to what happens in interaction nets). 
\item \textbf{$\parr$-boxes}: We put boxes on $\parr$-links and not on $!$-links. This choice is discussed in Section \ref{s:par}, and it allows to use a very simple correctness criterion---\ie\ Laurent's criterion for polarized nets---without losing any property. 
\item \textbf{Polarity}: we apply a polarized criterion to a setting which is not polarized in the usual sense. 
\item \textbf{Syntax tree}: since we use proof nets to represent terms, we will dispose them on the plane according to the syntax tree of the corresponding terms, and not according to the corresponding sequent calculus proof (also the orientation of the links does not reflect the usual premise-conclusion orientation of proof nets).
\end{enumerate}

\emph{Nets}. Nets are directed and labelled hyper-graphs $G=(V(G),L(G))$, \ie, graphs where $V(G)$ is a set of labelled \deff{nodes} and $L(G)$ is a set of labelled and \deff{directed hyperedges}, called \deff{links}, which are edges with 0,1 or more sources and 0,1 or more targets\footnote{
An hyper-graph $G$ can be understood as a bipartite graph $B_G$, where $V_1(B_G)$ is $V(G)$ and $V_2(B_G)$ is $L(G)$, and the edges are determined by the relations \textit{being a source} and \textit{being a target} of an hyperedge.}. Nodes are labelled with a type in $\set{e,m}$, where $e$ stays for \textit{exponential} and $m$ for \textit{multiplicative}, depicted in blue and brown, respectively. If a node $u$ has type $e$ (resp. $m$) we say that it is a $e$-node (resp. $m$-node). We shall consider hyper-graphs whose links are labelled from $\set{!,\der,\weak,\parr,\otimes}$. The label of a link $l$ forces the number and the type of the source and target nodes of $l$, as shown in Figure \ref{fig:links} (the types will be discussed later, and the figure also contains the $\square$-link, which is not used to define nets: it will be used later to define the correction graph). Note that every link (except $\square$) has exactly one connection with a little circle: it denotes the principal node, \ie\ the node on 
which the link can 
interact. Remark the principal node for tensor and $!$, which is not misplaced. Moreover, every $\parr$-link has an associated \deff{box}, \ie, a sub-hyper-graph of $G$ (have a look to Figure \ref{fig:trans}). The \deff{sources} (resp. \deff{targets}) of a net are the nodes without (resp. outgoing) incoming links; a node which is not a source nor a target is \deff{internal}. Formally:

\begin{figure}[t]
\centering
\ovalbox{\begin{tabular}{cccccccc|cccc}
\begin{tikzpicture}[ocenter]
\node at (0,0) [eport, label=below:$!m$](up){};

\lweakv{up}{weak}
\end{tikzpicture}
&
&
\begin{tikzpicture}[ocenter]
\node at (0,0) [label=above:$m$, mport](up){};
\node at (up.center) [below=\stalt, eport,label=below:$!m$](down){};
\lderv{down}{up}{der}
\end{tikzpicture}
&
&
\begin{tikzpicture}[ocenter]
\node at (0,0) [label=above:${e=!m}$, eport](up){};
\node at (up.center) [below=\stalt, label=below:$m$, mport](down){};
\lbangv{down}{up}{bang}
\end{tikzpicture}
&
\begin{tikzpicture}[ocenter]
\node at (0,0) [label=above:$e$, eport](up){};
\node at (up.center) [below right=\stalt and \hstlar,label=below:$e$,  eport](downr){};
\node at (up.center) [below left=\stalt and \hstlar,  label=below:${m=e^\bot\parr e}$, mport](downl){};
\ltensv{up}{downl}{downr}{tens}
\end{tikzpicture}
&

\begin{tikzpicture}[ocenter]
\node at (0,0) [ label=above:${m=e^\bot\parr e}$, mport](up){};
\node at (up.center) [below right=\stalt and \hstlar, label=below:$e$, eport](downr){};
\node at (up.center) [below left=\stalt and \hstlar, label=below:$e^\bot$, eport](downl){};
\lparv{up}{downl}{downr}{par}
\end{tikzpicture}
&
&&
\begin{tikzpicture}[ocenter]
\node at (0,0) [label=above:$m$, mport](up){};
\node at (up.center) [below right=\stalt and \hstlar, label=below:$e$, eport](downr){};
\node at (up.center) [below left=\stalt and \hstlar, label=below:$e$, eport](downl){};
\gdots{downl}{downr}{}
\lcolbox{up}{downl}{downr}{par}
\end{tikzpicture}

\end{tabular}
}
\caption{\label{fig:links} links.}
\end{figure}

\begin{definition}[net]
A \deff{net} $G$ is a quadruple $(|G|, B_G, \fv{G}, r_G)$, where $|G|=(V(G),L(G))$ is an hyper-graph  whose nodes are labelled with either $e$ or $m$ and whose hyperedges are $\set{!,\der,\weak,\parr,\otimes}$-links and s.t.:
\begin{itemize}
\item \deff{Root}: $r_G\in V(G)$ is a source $e$-node of $G$, called the \deff{root} of $G$.
\item \deff{Conclusions}: $\fv{G}$ is the set of targets of $G$, also called \deff{free variables} of $G$, which are targets of $\set{\der,\weak}$-links (and not of $\otimes$-links).
\item \deff{Multiplicative}: $m$-nodes have \emph{exactly one} incoming and \emph{one} outgoing link.
\item \deff{Exponential}: an $e$-node has at most one outgoing link, and if it is the target of more than one link then they all are $\der$-links. Moreover, an $e$-node cannot be isolated.
\item \deff{Boxes}: For every $\parr$-link $l$ there is a net $\bbox{l}$, called the \deff{box} of $l$ ($B_G$ is the set of boxes of $G$ and $\bbox{l}\in B_G$), with a distinguished free variable $x$, called the \deff{variable} of $l$, and s.t.:
\begin{itemize}
\item \deff{Border}: the root $r_{\bbox{l}}$ and the free variable $x$ are the $e$-nodes of $l$, and any free variable $\neq x$ of $\bbox{l}$ is not the target of a weakening.
\item \deff{Nesting}: for any two $\parr$-boxes $\bbox{l_1}$ and $\bbox{l_2}$ if $\emptyset\neq I\defeq\bbox{l_1}\cap\bbox{l_2}$, $\bbox{l_1}\not \subseteq\bbox{l_2}$, and $\bbox{l_2}\not \subseteq\bbox{l_1}$ then all the nodes in $I$ are free variables of both $\bbox{l_1}$ and $\bbox{l_2}$.
\item \deff{Internal closure}: any link $l$ of $G$ having as target an internal $e$-node of $\bbox{l}$ is in $\bbox{l}$.

\item \deff{Subnet}: the nodes and the links of $\bbox{l}$ belong to $G$ and the $\parr$-links in $\bbox{l}$ inherit the boxes from $G$.
\end{itemize}
\end{itemize}
\end{definition}

\emph{Some (technical) comments on the definition}. In the border condition the fact that the free variables $\neq x$ are not (the target) of a weakening means that weakenings are assumed to be pushed out of boxes as much as possible (of course the rewriting rules will have to preserve this invariant). The internal closure condition is a by-product of collapsing contractions on nodes, which is also the reason of the unusual formulation of the nesting condition: two boxes that are morally disjoint can in our syntax share free variables, because of an implicit contraction merging two of their conclusions.

\emph{Terminology about nets}. The \deff{level} of a node/link/box is the maximum number of nested boxes in which it is contained\footnote{Here the words \emph{maximum} and \emph{nested} are due to the fact that the conclusions of $\parr$-boxes may belong to two not nested boxes, because of the way we represent contraction.} (a $\parr$-link is not contained in its own box).  Two links are \deff{contracted} if they share an $e$-target. Note that the exponential condition states that only derelictions (\ie\ $\der$-links) can be contracted. In particular, no link can be contracted with a weakening. A \deff{free weakening} in a net $G$ is a weakening whose node is a free variable of $G$. Sometimes, the figures show a link in a box having as target a contracted $e$-node $x$ which is outside the box: in those cases $x$ is part of the box, it is outside of the box only in order to simplify the representation.

\emph{Typing}. Nets are typed using a recursive type $o=!(o\multimap o)$, that we rename $e=!(e\multimap e)=!(e^\bot\parr e)$ because $e$ is a mnemonic for \emph{exponential}. Let $m=e\multimap e=e^\bot\parr e$, where $m$ stays for  \emph{multiplicative}. Note that $e=!m$ and $m=!m\multimap !m$. Links are typed using $m$ and $e$, but the types are omitted by all figures except Figure \ref{fig:links} because they are represented using colors and with different shapes ($m$-nodes are brown and dot-like, $e$-nodes are white-filled cyan circles). Let us explain the types in Figure \ref{fig:links}. They have to be read bottom-up, and thus negated (to match the usual typing for links) if the conclusion of the logical rule is the bottom node of the link, as it is the case for the $\set{\weak,\der,\otimes}$-links, while $!$ and $\parr$ have their logical conclusion on the top node, and so their type does not need to be negated.

\begin{figure}
\centering
\ovalbox{\begin{tabular}{c|c|cccccccc|cccc}
a)
\begin{tikzpicture}[ocenter]
\node at (0,0) [eport](up){};
\inetcell[inductiveTr,at=(up.center), below=\sepbox](body){}[180]    
\node at (body.left pax) [eport, below=\sepbox](freevarleft){};
\node at (body.right pax) [eport, below=\sepbox](freevarright){};

\draw[exppure] (body.left pax) to (freevarleft);
\draw[exppure] (body.right pax) to (freevarright);
\draw[exppure] (up) to (body.pal);
\gdots{freevarleft}{freevarright}{}
\end{tikzpicture}
&
b)
\begin{tikzpicture}[ocenter]
\node at (0,0) [eport](up){};
\inetcell[inductiveTr,at=(up.center), below=\sepboxshort](body){}[180]    
\node at (body.left pax) [eport, below=\sepboxshort](freevarleft){};
\node at (body.middle pax) [eport, below=\sepboxshort](freevarmiddle){};
\node at (body.right pax) [eport, below=\sepboxshort](freevarright){};
\end{tikzpicture}
&
c)
\begin{tikzpicture}[ocenter]
\node at (0,0) [mport](up){};
\node at (up.center) [below right=\stalt and \hstlar, eport](downr){};

\inetcell[inductiveTr,at=(downr.center), below=\sepboxshort](body){\scriptsize $\lt{t}$}[180]    
\node at (body.left pax) [eport, below=\sepboxshort](x){};
\node at (body.middle pax) [eport, below=\sepbox](freevarmiddle){};
\node at (body.right pax) [eport, below=\sepbox](freevarright){};

\lparvanglefake{up}{x}{downr}{par}{out=180, in =225, looseness=1}

\aboxvnodes{par}{freevarright}{exbox}{15pt}{25pt}
\node at (body.middle pax) [eport, below=\sepbox](freevarmiddle){};
\node at (body.right pax) [eport, below=\sepbox](freevarright){};
\lparvangle{up}{x}{downr}{par}{out=180, in =225, looseness=1}
\draw[exppure] (body.middle pax) to (freevarmiddle);
\draw[exppure] (body.right pax) to (freevarright);
\end{tikzpicture}
$\Rew{\tt collapse}$
\begin{tikzpicture}[ocenter]
\node at (0,0) [mport](up){};
\node at (up.center) [below right=\stalt and \hstlar, eport](downr){};
\node at (up.center) [below left=\stalt and \hstlar, eport](downl){};
\gdots{downl}{downr}{}
\lcolbox{up}{downl}{downr}{par}
\end{tikzpicture}

\end{tabular}
}
\caption{\label{fig:various} various images.}
\end{figure}
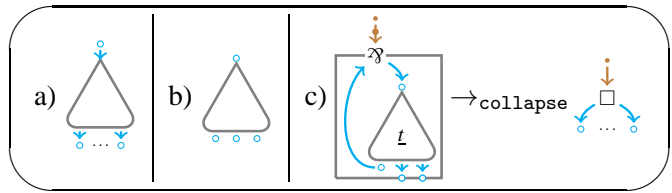

\emph{Induced $!$-boxes}. Note that a $!$-link is always applied to something ($m$-nodes cannot be conclusions), and there is not so much freedom for this \emph{something}: either it is a dereliction link or a $\parr$ with its box. Note also that in both cases we get (what would usually be) a valid content for a $!$-box. For the dereliction case it is evident, and for the $\parr$ case it is guaranteed by  the definition of net: the content of a $\parr$-box ends on $e$-nodes. Hence, any $!$-link has an associated box, induced by $\parr$-boxes, which needs not to be represented explicitly.

\emph{The translation}. Nets representing terms have the general form in Figure \ref{fig:various}.a, also schematized as in  Figure \ref{fig:various}.b. The translation $\lamtonets{\cdot }$ from terms to nets is in Figure \ref{fig:trans} (the original boring translation is sketched in Fig. \ref{fig:ord-trans}, page \pageref{fig:ord-trans}). A net which is the translation of a term is a \deff{proof net}. Note that in some cases there are various connections entering an $e$-node, that is the way we represent contraction. In some cases the $e$-nodes have an incoming connection with a perpendicular little bar: it represents an arbitrary number ($>0$) of incoming connections. The net corresponding to a variable is given by a $!$ on a dereliction and not by an (exponential) axiom, as it is sometimes the case. The reason is that an axiom (in our case a node, because axioms are collapsed on nodes) would not reflect on nets some term reductions, as $x[x/\val]\toe \val$, for which both the redex 
and the reduct would be mapped on the same net.

The translation $\lamtonets{\cdot}$ is refined to a translation $\lamtonetsvar{\cdot }{X}$, where $X$ is a set of variables, in order to properly handle weakenings during cut-elimination. The reason is that an erasing step on terms simply erases a subterm, while on nets it also introduces some weakenings: without the refinement the translation would not be stable by reduction. The clause defining $\lamtonetsvar{t}{X\cup\set{y}}$ when $y\notin\fv{t}$ is the first on the second line of Figure \ref{fig:trans}, the definition is then completed by the following two clauses: $\lamtonetsvar{t}{\emptyset}:= \lamtonets{t}$ and $\lamtonetsvar{t}{X\cup\set{y}}:= \lamtonetsvar{t}{X}$ if $y\in\fv{t}$.
 
\emph{$\alpha$-equivalence}. To circumvent an explicit and formal treatment of $\alpha$-equivalence we assume that the set of $e$-nodes and the set of variable names for terms coincide. This convention removes the need to label the targets of $\lamtonetsvar{t}{X}$ with the name of the corresponding free variables in $t$ or $X$. Actually, before translating a term $t$ it is necessary to pick a \emph{well-named} $\alpha$-equivalent term $t'$, \ie\ a term where any two different variables (bound or free) have different names.

\begin{figure}
\centering
\input{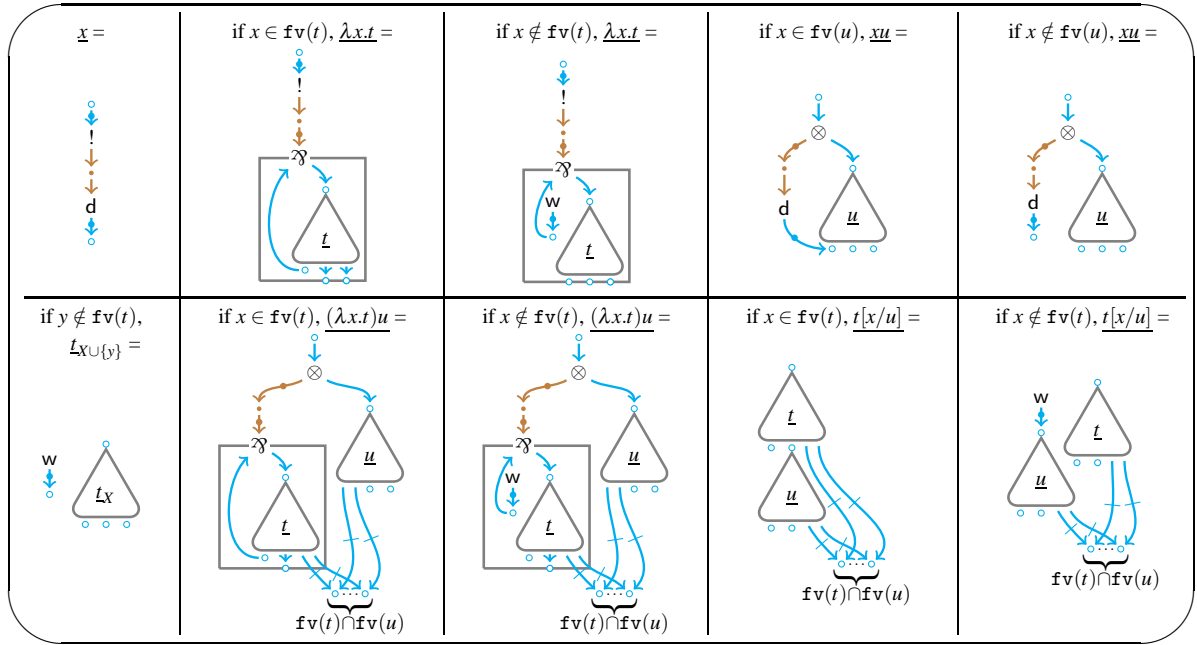}
\caption{\label{fig:trans} the translation from terms to nets.}
\end{figure}

\begin{remark}
\label{rem:not-inj}
The translation of terms to nets is not injective. By simply applying the translation it is easily seen that the following pairs of terms have the same net:
\begin{equation}
\label{eq:quotient}
 \begin{array}{lll@{\hspace{.5cm}}l}
t[x/s][y/u] & \preeqw{\vosym_\CS} & t[y/u][x/s] & \mbox{if } x\notin\fv{u} \ \&\ y\notin\fv{s}\\
\val\ u[x/s] & \preeqw{\vosym_1} & (\val\ u)[x/s] & \mbox{if } x\notin\fv{\val}\\
t[x/s[y/u]] & \preeqw{\vosym_2} & t[x/s][y/u] & \mbox{if } y\notin\fv{t}\\
\end{array}\end{equation}
Let $\eqvo$ be the reflexive, transitive, and contextual closure of $\preeqw{\vosym_\CS}\cup \preeqw{\vosym_1}\cup \preeqw{\vosym_2}$. In the proof of Lemma \ref{l:graph-sub}, we will use the fact that if $t\eqvo s$ then $t$ and $s$ are mapped on the same net. We also claim---without proving it---that $\eqvo$ is exactly the quotient induced on terms by the translation to nets. 
\end{remark}

\emph{Paths.} A path $\tau$ of length $k\in\nat$ from $u$ to $v$, noted $\tau:u\to^k v$, is an alternated sequence $u=u_1,l_1,\ldots,l_k, u_{k+1}=v$ of nodes and links s.t. the link $l_i$ has source $u_i$ and target $u_{i+1}$ for $i\in\set{1,\ldots,k}$. A cycle is a path $u\to^k u$ with $k>0$.

\emph{Correctness}. The correctness criterion is based on the notion of correction graph, which is---as usual for nets with boxes---obtained by collapsing every box at level 0 into a generalized axiom link.

\begin{definition}[correction graph]
Let $G$ be a net. The correction graph $\zeronet{G}$ of $G$ is the hyper-graph obtained from $G$ by collapsing any $\parr$-box at level 0 into a $\square$-link applying the rule in Fig. \ref{fig:various}.c. 
\end{definition}


\begin{definition}[correctness]
A net $G$ is correct if:
\begin{itemize}
\item \deff{Source}: $\zeronet{G}$ has exactly one $e$-source (the root of $G$).
\item \deff{Acyclicity}: $\zeronet{G}$ is acyclic.
\item \deff{Recursive correctness}: the interior of every box is correct.
\end{itemize}
\end{definition}

As usual an easy induction on the translation shows that the translation of a term is correct, \ie\ that:

\begin{lemma}
Every proof net is correct.
\end{lemma}
\section{Proof nets: sequentialization}
In this section we show how to extract a term $t$ from every correct net $G$ in such a way that $t$ translates back to $G$, \ie\ we show that every correct net is a proof net. The proof of this fact is based on the notion of \emph{kingdom}, along the lines of the proof for polarized nets, see \cite{phdlaurent} (pp. 57-63).

\begin{definition}[Kingdom]
\label{d:king}
Let $G$ be a correct net and $x\notin\fv{G}$ one of its $e$-nodes. The \deff{kingdom} $\kingv{x}$  of $x$ is the set of links defined by induction on the link $l$ of source $x$:
\begin{itemize}
\item $l$ is a $!$-link: $\kingv{x}$ is given by $l$ plus the $\der$-link or the $\parr$-box on the $m$-target of $l$.
\item $l$ is a $\otimes$-link: $\kingv{x}$ is given by $l$ plus the $\der$-link or the $\parr$-box on the $m$-target of $l$ plus $\kingv{y}$, where $y$ is the $e$-target of $l$.
\end{itemize}
\end{definition}

The main property of $\kingv{x}$ is that it is the smallest subnet of root $x$, as we shall soon prove\footnote{We call \emph{kingdom of $x$} the net in def. \ref{d:king}, but at this point nothing guarantees that it is the smallest subnet of root $x$.}. To state this fact precisely we need the notion of subnet.

\begin{definition}[subnet]
Let $G$ be a correct net. A subnet $H$ of $G$ is a subset of its links s.t. it is a correct net and satisfying:
\begin{itemize}
\item \deff{Internal closure}: if $x$ is an internal $e$-node of $H$ then any link of $G$ of target $x$ belongs to $H$.
\item \deff{Box closure}: 
\begin{itemize}
\item \deff{Root}: if a $\parr$-link $l$ belongs to $H$ then its box does it too. 
\item \deff{Free variables}: if a free variable of a box $B$ of $G$ is internal to $H$ then $B\subseteq H$.
\end{itemize}
\end{itemize}
\end{definition}

The following lemma is essentially obvious, and usually omitted, but in fact it is used in the proof of Lemma \ref{l:kingdom}.

\begin{lemma}
\label{l:subsubnets}
Let $G$ be a correct net, $H$ a subnet of $G$, $x$ an internal $e$-node of $H$. Then there exists a subnet $K$ of $H$ having $x$ as root and s.t. it is a subnet of $G$.
\end{lemma}

\begin{proof}
It is enough to show that there is a subnet of $H$ of root $x$, since it is obvious that any subnet of $K$ is a subnet of $G$. By induction on the length of the maximum path from $x$ to a free variable of $K$.
\end{proof}

To properly describe kingdoms we need the following definition.

\begin{definition}[(free/ground) substitution]
Let $G$ be a correct net. A \deff{substitution} is an $e$-node which is the target of a $\set{\weak,\der}$-link (or, equivalently, which is not the target of a $\otimes$-link) and the source of some link. A substitution $x$ is \deff{ground} if it is a node of $\zeronet{G}$ (\ie\ it is not internal to any $\parr$-box\footnote{Note that our collapsed representation of contractions and cuts does not allow to simply say that $x$ is a node at level 0: indeed the conclusion of a $\parr$-box can have level $>0$ and yet belong to $\zeronet{G}$.}), and it is \deff{free} if it is ground and there is no ground substitution of $G$ to which $x$ has a path (in $\zeronet{G}$).
\end{definition}

\begin{lemma}[kingdom]
\label{l:kingdom}
Let $G$ be a correct net and $x\notin\fv{G}$ one of its $e$-nodes. $\kingv{x}$ is the kingdom of $x$, \ie, the smallest subnet of $G$ rooted at $x$. Moreover, it has no free substitutions, no free weakenings, and whenever $y\in\fv{\kingv{x}}$ is internal to a subnet $H$ of $G$ then $\kingv{x}\subseteq H$.
\end{lemma}

\proof
Let $H$ be a correct subnet of $G$ rooted at $x$. We show by induction on the length of the maximum path from $x$ to a free variable of $G$ that $\kingv{x}\subseteq H$ and that $\kingv{x}$ is correct. Let $l$ be the link of source $x$. Cases:
\begin{itemize}
\item \deff{Base case}: $l$ is a $!$-link. By the conclusion condition $H$ has to contain the $\der$-link $i$ or the $\parr$-link on the $m$-target of $l$. In the case of a $\parr$-link the box closure condition implies that the whole box $B$ is in $H$, hence $\kingv{x}\subseteq H$. In the case of a $\der$-link correctness is obvious, in the case of a $\parr$-box it follows by the correctness of the interior of the box, guaranteed by the recursive correctness condition. Moreover, no free substitutions and no free weakenings belong to $\kingv{x}$ (boxes cannot close on weakenings). Pick $y\in\fv{\kingv{x}}$, which in the $\der$-link case is the target of $i$ and in the other case is a free variable of the $\parr$-box $B$. If $y$ is internal to $H$ then the conditions for a subnet guarantee that $i$ or $B$ are in $H$. Then clearly $\kingv{x}\subseteq H$.

\item \deff{Inductive case}: $l$ is a $\otimes$-link. As in the previous case $H$ has to contain the $\der$-link or the $\parr$-box on the $m$-target of $l$. Moreover, by lemma \ref{l:subsubnets} $H$ contains a subnet $K$ rooted in the $e$-target $y$ of $l$. By inductive hypothesis $\kingv{y}$ is the kingdom of $y$, therefore we get $\kingv{y}\subseteq K\subseteq H$. Hence $\kingv{x}\subseteq H$. By \ih\ we also get that $\kingv{y}$ is correct, hence $y$ is its only $e$-source and  $x$ is the only $e$-source of $\kingv{x}$. Acyclicity follows by correctness of $G$. Recursive correctness follows from the box closure condition and correctness of $G$. Moreover, by \ih\ $\kingv{y}$---and so $\kingv{x}$---has no free substitutions and no free weakenings. The part about free variables uses the \ih\ for the free variables of $\kingv{y}$ and the conditions for a subnet as in the previous case for the other free variables.
\qed
\end{itemize}

\begin{lemma}[substitution splitting]
\label{l:sub-splitting}
Let $G$ be a correct net with a free substitution $x$. Then 
\begin{enumerate}
\item The free variables of $\kingv{x}$ are free variables of $G$.
\item $G\setminus \kingv{x}$ is a subnet of $G$.
\end{enumerate}
\end{lemma}

\begin{proof}
1) Suppose not. Then there is a free variable $y$ of $\kingv{x}$ which is not a free variable of $G$. There are two possible cases:
\begin{itemize}
\item \emph{$y$ is a substitution}. Then $x$ has a path to a substitution in $\zeronet{G}$, against the definition of free substitution, absurd.
\item \emph{$y$ is the distinguished free variable of a $\parr$-box $B$}. Thus, $y$ is internal to some $\parr$-box $B$ and so it is not a node of $\zeronet{G}$. By Lemma \ref{l:kingdom} we get that $\kingv{x}\subseteq B$ and so $x$ is not a node of $\zeronet{G}$, against the definition of free substitution, absurd.

\end{itemize}
2) By point 1 the removal of $\kingv{x}$ cannot create new $e$-sources. Being a substitution, $x$ is the target of some link. Therefore the removal of $\kingv{x}$ cannot remove the root of $G$. It is also clear that the removal cannot create cycles, and the box closure condition for subnets guarantees that the recursive correctness of $G$ implies the one of $G\setminus \kingv{x}$.
\end{proof}

\begin{lemma}
\label{l:no-free-sub-no-sub}
Let $G$ be a correct net with a ground substitution. Then $G$ has a free substitution.
\end{lemma}

\begin{proof}
Consider the following order on the elements of the set $S_g$ of ground substitutions of $G$: $z\leq y$ if there is a path from $z$ to $y$ in $\zeronet{G}$. Acyclicity of $\zeronet{G}$ implies that $S_g$ contains maximal elements with respect to $\leq$, if it is non-empty. Note that a maximal element of $S_g$  is a free substitution in $G$. Now, if $G$ has a ground substitution $x$ then $S_g$ is non-empty. Thus, $G$ has a free substitution.
\end{proof}

The next lemma is used in the proof of the sequentialization theorem.

\begin{lemma}[kingdom characterization]
\label{l:no-free-sub}
Let $G$ be a correct net. Then $G= \kingv{r_G}$ iff $G$ has no free substitutions nor free weakenings.
\end{lemma}

\begin{proof}
$\Rightarrow$) By Lemma \ref{l:kingdom}. $\Leftarrow$) By lemma \ref{l:kingdom} we get that $\kingv{r_G}\subseteq G$. If the two do not coincide then by the internal closure condition for subnets, the multiplicative condition on nets, and the fact that they share the same root, we get that $G$ contains a ground substitution $x$ on a free variable of $\kingv{r_G}$. By lemma \ref{l:no-free-sub-no-sub} $G$ contains a free substitution, absurd.
\end{proof}

\begin{theorem}[sequentialization]
Let $G$ be a correct net and $X$ be the set of $e$-nodes of its free weakenings. Then there is a term $\tm$ s.t. $\lamtonetsvar{\tm}{X}=G$ (and $\fv{G}=\fv{\tm}\cup X$).
\end{theorem}

\proof
By induction on the number of links. By the root and conclusion conditions the minimum number of links is 2 and the two links are necessarily a $!$-link on top of a $\der$-link. Let $x$ be the $e$-node of the $\der$-link. Then $\lamtonets{x}=G$. We now present each inductive case. After the first one we assume that the net has no free weakening.
\begin{itemize}
\item \textit{There is a free weakening $l$ of $e$-node $y$}. Then $G'=G\setminus \set{l}$ is still a correct net and by \ih\ there exist $t$ s.t. $\lamtonetsvar{\tm}{X\setminus\set{y}}=G'$. Then $\lamtonetsvar{\tm}{X}=G$.
\item \textit{There is a free substitution $x$}. Then by Lemma \ref{l:kingdom} and Lemma \ref{l:sub-splitting} $\kingv{x}$ and $G\setminus\kingv{x}$ are correct subnets of $G$. By the \ih\ there exist $\tmtwo$ and $\tmthree$ s.t. $\lamtonets{\tmtwo}=\kingv{x}$ and $\lamtonetsvar{\tmthree}{\set{x}}=G\setminus\kingv{x}$ (note that if $x\in \fv{\tmthree}$ then $\lamtonetsvar{\tmthree}{\set{x}}=\lamtonetsvar{\tmthree}{\emptyset}=\lamtonets{\tmthree}$). Then $\lamtonets{\tmthree\esub{\var}{\tmtwo}}=G$.
\item \textit{No free substitution}: by lemma \ref{l:no-free-sub} $G=\kingv{r_G}$. In case the root link $l$ of $G$ is:
\begin{itemize}
\item \textit{a $!$-link over a $\der$-link}: base case, already treated.
\item \textit{a $!$-link over a $\parr$-link}: let $H$ be the box of the $\parr$-link and $x$ its distinguished free variable. By definition of a net the set of free weakenings of $H$ either is empty or it contains only $x$. If $x$ is (resp. is not) the node of a free weakening then by \ih\ there exists $t$ s.t. $\lamtonetsvar{t}{\set{x}}=H$ (resp. $\lamtonets{t}=H$). Then $\lamtonets{\l x.t}=G$.
\item \textit{A $\otimes$-link $l$}: let $x$ be its $e$-target and $a$ its $m$-target. Note that $G=\kingv{r_G}$ implies that $G$ is composed by $l$, $\kingv{x}$ and either the $\der$-link or the $\parr$-link (plus its box) on $a$. By \ih\ there exists $\tmtwo$ s.t. $\lamtonets{\tmtwo}=\kingv{x}$. Now, if $a$ is the source of a $\der$-link of $e$-node $y$ we conclude, since $\lamtonets{y \tmtwo}= G$. Otherwise, $s$ is the source of a $\parr$ of box $H$ and the \ih\ gives a term $\tmthree$ and a set $X$ s.t. $\lamtonetsvar{\tmthree}{X}=H$. Let us prove that $H$ and $\kingv{x}$ can only share free variables, as the translation prescribes: no link at level $0$ of $\kingv{x}$ can be in $H$, and no box at level 0 of $\kingv{x}$ can intersect $H$ other than on free variables, by the nesting condition. By reasoning about the distinguished free variable of $H$ as in the previous case we then get $\lamtonets{(\l y. \tmthree) \tmtwo}=G$.
\qed
\end{itemize}
\end{itemize}

\section{Proof nets: dynamics}
\label{s:pn-dynamics}
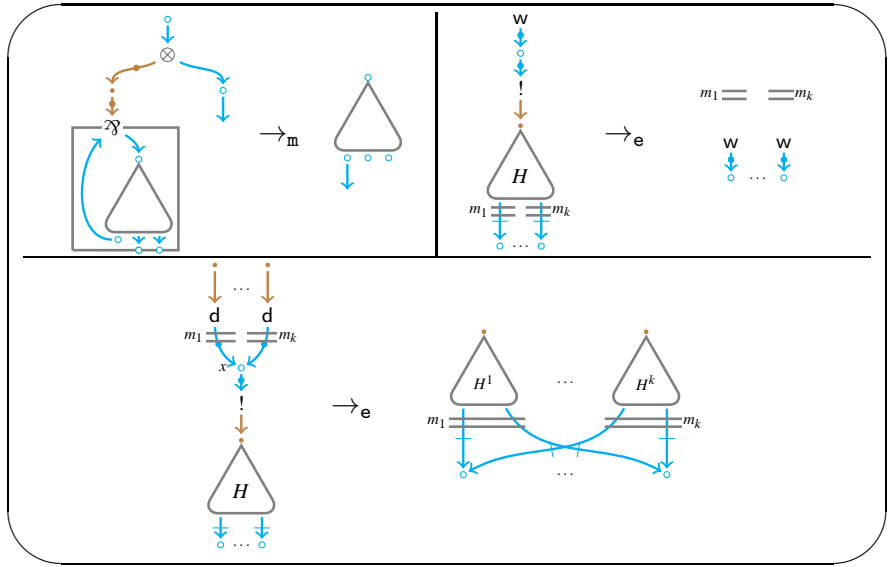
\begin{figure}[t] 
\centering\ovalbox{
\begin{tabular}{cccc}
\begin{tabular}{c|ccc}
\begin{tabular}{ccc}
\begin{tikzpicture}[vcenter]
\node [eport] (out){};
\node at (out.center) [mport, below left=\stalt and \stlar] (fun){};
\node at (out.center) [eport, below right=\stalt and \stlar] (arg){};
\ltensv{out}{fun}{arg}{ap};

\node at (fun.center) [below right=\stalt and \hstlar, eport](downr){};

\inetcell[inductiveTr,at=(downr.center), below=\sepboxshort](body){}[180]    
\node at (body.left pax) [eport, below=\sepboxshort](x){};
\node at (body.middle pax) [eport, below=\sepbox](freevarmiddle){};

\lparvanglefake{fun}{x}{downr}{par}{out=180, in =225, looseness=1}
\aboxvnodes{par}{freevarmiddle}{exbox}{15pt}{25pt}

\node at (body.middle pax) [eport, below=\sepbox](freevarmiddle){};

\lparvangle{fun}{x}{downr}{par}{out=180, in =225, looseness=1}

\node at (arg.center)[nospace, below=\hstalt](ghost){};
\draw[exppure] (arg) to (ghost);

\node at (body.right pax) [eport, below=\sepbox](freevarright){};
\draw[exppure] (body.middle pax) to (freevarmiddle);
\draw[exppure] (body.right pax) to (freevarright);
\end{tikzpicture}
&
$\tom$
&

\begin{tikzpicture}[vcenter]
\node at (0,0)[eport](out){};
\inetcell[inductiveTr,at=(out.center),below=\sepboxshort](argbody){}[180]    

\node at (argbody.middle pax)[eport, below=\sepboxshort](argfreevar2){};
\node at (argbody.left pax)[eport, below=\sepboxshort](vargluarg){};

\node at (argbody.right pax) [eport, below =\sepboxshort](freevarright1){};

\node at (vargluarg.center)[nospace, below=\hstalt](ghost){};
\draw[exppure] (vargluarg) to (ghost);

\end{tikzpicture}
\end{tabular}&\begin{tabular}{ccc}
\begin{tikzpicture}[vcenter]
\node at (0,0) [eport] (input){};
\lweakv{input}{weakening};
\node at (input.center) [mport, below =\stalt] (beta){};
\lbangv{beta}{input}{bang}

\inetcell[inductiveTr,at=(beta.center), below=1pt](argbody){}[180]
\node at (argbody.center) [nospace] (z1){\scriptsize $H$};
\node at (argbody.left pax) [below=3.5*\sepbox,eport] (z1){};
\node at (argbody.right pax) [below=3.5*\sepbox,eport] (zn){};
\gdots{z1}{zn}{}

\draw[genexppure, exppure](argbody.right pax) to (zn);
\draw[genexppure, exppure](argbody.left pax) to (z1);

\node at (argbody.left pax) [below left=2pt and 0.1*\stlar](inlub){};
\node at (argbody.left pax) [below right=2pt and 0.1*\stlar](inrub){};
\node at (argbody.left pax) [below left=-1pt and 0.1*\stlar](inldb){};
\node at (argbody.left pax) [below right=-1pt and 0.1*\stlar](inrdb){};
\draw[exboxline](inlub.center)--(inrub.center);
\draw[exboxline](inldb.center)--(inrdb.center);
\node at \med{inlub}{inldb}[left=2pt, etic](m1){\tiny $m_1$};

\node at (argbody.right pax) [below left=2pt and 0.1*\stlar](inlubk){};
\node at (argbody.right pax) [below right=2pt and 0.1*\stlar](inrubk){};
\node at (argbody.right pax) [below left=-1pt and 0.1*\stlar](inldbk){};
\node at (argbody.right pax) [below right=-1pt and 0.1*\stlar](inrdbk){};
\draw[exboxline](inlubk.center)--(inrubk.center);
\draw[exboxline](inldbk.center)--(inrdbk.center);
\node at \med{inrubk}{inrdbk}[right=2pt, etic](mk){\tiny $m_k$};
\end{tikzpicture}
&
$\toe$
&
\boxinit
\begin{tikzpicture}[vcenter]
\node at (0,0) [eport] (input){};
\lweakv{input}{weakening};

\node at (input.center) [right=\stlar, eport] (input2){};
\lweakv{input2}{weakening2};

\gdots{input}{input2}{}

\node at (input) [above left=\stalt+3pt and 0.1*\stlar](inlub){};
\node at (input) [above right=\stalt+3pt and 0.1*\stlar](inrub){};
\node at (input) [above left=\stalt and 0.1*\stlar](inldb){};
\node at (input) [above right=\stalt and 0.1*\stlar](inrdb){};
\draw[exboxline](inlub.center)--(inrub.center);
\draw[exboxline](inldb.center)--(inrdb.center);
\node at \med{inlub}{inldb}[left=2pt, etic](m1){\tiny $m_1$};

\node at (input2) [above left=\stalt+3pt and 0.1*\stlar](inlubk){};
\node at (input2) [above right=\stalt+3pt and 0.1*\stlar](inrubk){};
\node at (input2) [above left=\stalt and 0.1*\stlar](inldbk){};
\node at (input2) [above right=\stalt and 0.1*\stlar](inrdbk){};
\draw[exboxline](inlubk.center)--(inrubk.center);
\draw[exboxline](inldbk.center)--(inrdbk.center);
\node at \med{inrubk}{inrdbk}[right=2pt, etic](mk){\tiny $m_k$};

\end{tikzpicture}
\end{tabular}
\end{tabular}
\\\hline
\begin{tabular}{ccc}
\begin{tikzpicture}[vcenter]
\node at (0,0) [mport] (o1){};
\node at (o1) [dummynode, above=\sepbox/3](spacingNode){};
\node at (o1.center) [ right= \stlar,  mport] (ok){};
\gdots{o1}{ok}{below=\sepbox}

\node at \med{o1}{ok} [below=1.5*\stalt, label=left: $x$, eport] (input){};
\ldervangle{input}{o1}{varoc1}{out=-89,in=150}
\ldervangle{input}{ok}{varock}{out=-90,in=30}

\node at (input.center) [mport, below =\stalt] (beta){};
\lbangv{beta}{input}{bang}

\inetcell[inductiveTr,at=(beta.center), below=1pt](argbody){}[180]
\node at (argbody.center) [nospace] (z1){\scriptsize $H$};
\node at (argbody.left pax) [below=2.3*\sepbox,eport] (z1){};
\node at (argbody.right pax) [below=2.3*\sepbox,eport] (zn){};
\gdots{z1}{zn}{}

\draw[genexppure, exppure](argbody.right pax) to (zn);
\draw[genexppure, exppure](argbody.left pax) to (z1);

\node at (varoc1.center) [below left=6pt and 0.1*\stlar](inlub){};
\node at (varoc1.center) [below right=6pt and 0.2*\stlar](inrub){};
\node at (varoc1.center) [below left=3pt and 0.1*\stlar](inldb){};
\node at (varoc1.center) [below right=3pt and 0.2*\stlar](inrdb){};
\draw[exboxline](inlub.center)--(inrub.center);
\draw[exboxline](inldb.center)--(inrdb.center);
\node at \med{inlub}{inldb}[left=2pt, etic](m1){\tiny $m_1$};

\node at (varock.center) [below left=6pt and 0.2*\stlar](inlubk){};
\node at (varock.center) [below right=6pt and 0.1*\stlar](inrubk){};
\node at (varock.center) [below left=3pt and 0.2*\stlar](inldbk){};
\node at (varock.center) [below right=3pt and 0.1*\stlar](inrdbk){};
\draw[exboxline](inlubk.center)--(inrubk.center);
\draw[exboxline](inldbk.center)--(inrdbk.center);
\node at \med{inrubk}{inrdbk}[right=2pt, etic](mk){\tiny $m_k$};

\end{tikzpicture}
&
$\toe$
&
\boxinit
\begin{tikzpicture}[vcenter]

\node at (0,0) [mport] (beta){};

\node at (beta.center) [mport, right =3*\stlar] (beta2){};

\inetcell[inductiveTr,at=(beta.center), below= 1pt](arg1){}[180]    
\node at (arg1.center) [nospace] (e1){\tiny $H^1$};

\inetcell[inductiveTr,at=(beta2.center), below= 1pt](argk){}[180]    
\node at (argk.center) [nospace] (e2){\tiny $H^k$};
\gdots{arg1}{argk}{}

\node at (arg1.left pax) [below =\stalt, eport](z1){};
\draw[genexppure, exppure](arg1.left pax) to (z1);
\draw[genexppure, exppure, out=-120, in =30](argk.left pax) to (z1);

\node at (argk.right pax) [below =\stalt, eport](zn){};
\draw[genexppure, exppure, out=-60, in =150](arg1.right pax) to (zn);
\draw[genexppure, exppure](argk.right pax) to (zn);

\gdots{z1}{zn}{}

\node at (arg1.left pax) [below left=4pt and 0.2*\stlar](inlub){};
\node at (arg1.right pax) [below right=4pt and 0.2*\stlar](inrub){};
\node at (arg1.left pax) [below left=1pt and 0.2*\stlar](inldb){};
\node at (arg1.right pax) [below right=1pt and 0.2*\stlar](inrdb){};
\draw[exboxline](inlub.center)--(inrub.center);
\draw[exboxline](inldb.center)--(inrdb.center);
\node at \med{inlub}{inldb}[left=2pt, etic](m1){\tiny $m_1$};

\node at (argk.left pax) [below left=4pt and 0.2*\stlar](inlubkk){};
\node at (argk.right pax) [below right=4pt and 0.2*\stlar](inrubkk){};
\node at (argk.left pax) [below left=1pt and 0.2*\stlar](inldbkk){};
\node at (argk.right pax) [below right=1pt and 0.2*\stlar](inrdbkk){};
\draw[exboxline](inlubkk.center)--(inrubkk.center);
\draw[exboxline](inldbkk.center)--(inrdbkk.center);
\node at \med{inrubkk}{inrdbkk}[right=2pt, etic](mk){\tiny $m_k$};

\end{tikzpicture}
\end{tabular}
\end{tabular}
}
\caption{proof nets cut-elimination rules\label{fig:rew-rules}}
\end{figure}

The rewriting rules are in Figure \ref{fig:rew-rules}. Let us explain them. First of all, note that the notion of cut in our syntax is implicit, because cut-links are not represented explicitly. A cut is given by a node whose incoming and outgoing connections are principal (\ie\ with a little square on the line). 

The rule $\tom$ is nothing but the usual elimination of a multiplicative cut, except that the step also opens the box associated with the $\parr$-link.

The two $\toe$ rules reduce the exponential redexes. Let us explain how to read them. For the graph noted $H$ in Figure \ref{fig:rew-rules} there are two possibilities: either it is simply a dereliction link (a $\der$-link) or it is a $\parr$ with its box, so there is no ambiguity on what to duplicate/erase. Every pair of short gray lines denotes the sequence (of length $m_i$, with $i\in\set{1,\ldots, k}$) of boxes closing on the corresponding links. The rule has two cases, one where $!$ is cut with $k\in\set{1,2,\ldots}$ derelictions and one where it is cut with a weakening. In the first case the sub-graph $H$ is copied $k$ times (if $k=1$ no copy is done) into $H^1,\ldots H^k$ and each copy enters in the $m_i$ boxes enclosing the corresponding (and removed) dereliction. Moreover, the $k$ copies of each target of $H$ are contracted together, \ie\ the nodes are merged. In the case of a cut with a weakening, $H$ is erased and replaced by a set of weakenings, one for every target of $H$. Note that the weakenings are also pushed out of all boxes closing on the targets of $H$\footnote{Note that, for the sake of a simple representation, the figure of the weakening cut-elimination rule is slightly wrong: it is not true that the links $l_1,\ldots,l_j$ having as target a given conclusion $x_i$ of $H$ are all inside $m_i$ boxes, because each one can be inside a different number of boxes.}. This is done to preserve the invariant that weakening are always pushed out of boxes as much as possible. Such invariant is also used  in the rule: the weakening is at the same level of $H$. Last, if the weakenings created by the rule are contracted with any other link then they are removed on the fly (because by definition weakenings cannot be contracted).

Now, we establish the relationship between terms and nets at the level of reduction. Essentially, there is only one fact which is not immediate, namely that $\toe$ actually implements the $\toe$ rule on terms, as it is proved by the following lemma.

\begin{lemma}[substitution]
\label{l:graph-sub}
Let $t=s[x/\val\L]$ then $\lamtonetsvar{t}{X}\toe\lamtonetsvar{s\set{x/\val}\L}{X}$ for any set of names $X\supseteq \fv{t}$.
\end{lemma}

\begin{proof}
First of all observe that $t$ and $s[x/\val]L$ both reduce to $s\set{x/\val}L$ and by remark \ref{rem:not-inj} both translate to the same net. Hence it is enough to prove that $\lamtonetsvar{s[x/\val]L}{X}\Rew{e}\lamtonetsvar{s\set{x/\val}L}{X}$. We prove it by induction on the number $k$ of substitutions in $L$. If $k=0$ then the proof is by induction on the number $n$ of free occurrences of $x$ in $s$. Cases:
\begin{itemize}
\item $n=0$) In $\lamtonetsvar{s[x/\val]}{X}$ the bang associated to $\val$ is cut with a weakening. The elimination of the cut gets a net $G'$ without the $!$-link and the $\parr$-box associated to $\val$, leaving a free weakening for every free variable of the box, \ie\ of every free variable of $\val$: then $G'$ is exactly $\lamtonetsvar{s\set{x/\val}}{X\cup\fv{\val}}=\lamtonetsvar{s}{X\cup\fv{\val}}$.

\item $n>1$) Write $s=C[x]$ for some occurrence of $x$. Now, consider  $\tmthree=C[y][y/\val][x/\val]$ and note that:
\begin{center}
$\tmthree\Rew{}C[\val][x/\val]\Rew{}C[\val]\set{x/\val}=s\set{x/\val}$
\end{center}
The difference between $G'=\lamtonetsvar{\tmthree}{X}$ and $G=\lamtonetsvar{s[x/\val]}{X}$ is that one of the occurrences of $x$ in $G$ has been separated from the others and cut with a copy of $\lamtonets{\val}$. Consider the step $G\Rew{}H$ which reduces the cut on $x$ in $G$ and the sequence $G'\Rew{} H'_y \Rew{} H'_{y,x}$ which first reduces the cut on $y$ in $G'$ and then reduces in $H'$ the (unique) residual of the cut on $x$ in $G'$. By the definition of reduction in nets $H=H'_{y,x}$. Now by \ih\ applied to $\tmthree$ and $y$ we get that $\lamtonetsvar{C[\val][x/\val]}{X}=H'_y$ and by the \ih\ applied to $C[\val][x/\val]$ and $x$ we get that $\lamtonetsvar{C[\val]\set{x/\val}}{X}=H'_{y,x}$. From $H=H'_{y,x}$ and $C[\val]\set{x/\val}=s\set{x/\val}$ we get $\lamtonetsvar{s\set{x/\val}}{X}=H$ and conclude.

\item $n= 1$) By induction on $s$. Some cases:
\begin{itemize}
\item If $t= \l y. u$ then by \ih\ $\lamtonetsvar{u[x/\val]}{X\cup\set{y}}\Rew{e}\lamtonetsvar{u\set{x/\val}}{X\cup\set{y}}$ and so we get $\lamtonetsvar{\l y. (u[x/\val])}{X\cup\set{y}}\Rew{e}\lamtonetsvar{\l y. (u\set{x/\val})}{X\cup\set{y}}$. Now, observe that $\l y. (u\set{x/\val})=(\l y. u)\set{x/\val}=t\set{x/\val}$ and that the two nets $\lamtonetsvar{\l y. (u[x/\val])}{X\cup\set{y}}$ and $\lamtonetsvar{(\l y. u)[x/\val]}{X\cup\set{y}}$ have the same reduct after firing the exponential cut on $x$, and so we get $\lamtonetsvar{(\l y. u)[x/\val]}{X\cup\set{y}}\Rew{e}\lamtonetsvar{(\l y. u)\set{x/\val})}{X\cup\set{y}}$.
\item If $s=w[y/u]$ then either $x\in u$ or $x\in w$. In the first case by remark \ref{rem:not-inj} we get that $\lamtonetsvar{s[x/\val]}{X}=\lamtonetsvar{w[y/u][x/\val]}{X}=\lamtonetsvar{w[y/u[x/\val]]}{X}$. Now by \ih\ $\lamtonets{u[x/\val]}\Rew{e} \lamtonets{u\set{x/\val}}$. Then we have $\lamtonetsvar{s[x/\val]}{X}\Rew{e} \lamtonetsvar{w[y/u\set{x/\val}]}{X}=\lamtonetsvar{w[y/u]\set{x/\val}}{X}=\lamtonetsvar{s\set{x/\val}}{X}$. The second case is analogous.
\item If $s= (\l y.w)u$. The case $x\in u$ uses remark \ref{rem:not-inj} and the \ih\ as in the $s=w[y/u]$ case. The case $x\in w$ is slightly different. As before $((\l y.w)u)[x/\val]$ and $((\l y.w[x/\val])u)$ have the same reduct. By \ih\ hypothesis $\lamtonets{w[x/\val]}\Rew{e}\lamtonets{w\set{x/\val}}$ and thus $\lamtonetsvar{(\l y.w[x/\val])u}{X}\Rew{e}\lamtonetsvar{(\l y.w\set{x/\val})u}{X}$. We conclude since $\lamtonetsvar{((\l y.w)u)[x/\val]}{X}\Rew{e}\lamtonetsvar{((\l y.w\set{x/\val})u)}{X}=\lamtonetsvar{((\l y.w)u)\set{x/\val}}{X}$.
\end{itemize}

\end{itemize}
If $k>0$ and $\L=\L'[y/r]$ then we get by \ih\ that $\lamtonetsvar{s[x/\val]\L'}{X}\Rew{e}\lamtonetsvar{s\set{x/\val}\L'}{X}$. By definition of the translation and of graph reduction it follows that $\lamtonetsvar{s[x/\val]\L'[y/r]}{X}\Rew{e}\lamtonetsvar{s\set{x/\val}\L'[y/r]}{X}$.
\end{proof}

\begin{theorem}[strong bisimulation]
\label{tm:str-bis}
Let $t$ be a term and $X$ a set of variables containing $\fv{t}$. The translation is a strong bisimulation between $t$ and $\lamtonetsvar{t}{X}$, \ie\ $t\Rew{a} t'$ if and only if $\lamtonetsvar{t}{X}\Rew{a} \lamtonetsvar{t'}{X}$, for $a\in\set{\msym,\esym}$.
\end{theorem}

\begin{proof}
By induction on the translation. If $t=x$ there is nothing to prove, and if $t=\l x.s$ or $t=xs$ it immediately follows by the \ih, since all the redexes of $t$ are contained in $s$. If $t=s[x/u]$ and the redex is in $s$ or $u$ then just apply the \ih. If $u=\val\L$ and the redex is $s[x/\val\L]\toe s\isub{x}{\val}\L$  then apply Lemma \ref{l:graph-sub}. If $t=(\l x. s)u$ and the redex is in $s$ or $u$ then just apply the \ih. If $t=(\l x. s)u\tom s[x/u]=t'$ then have a look at Figure \ref{fig:counter}.a: clearly $t\tom t'$ iff $\lamtonetsvar{t}{X}\tom \lamtonetsvar{t'}{X}$.
\end{proof}

\begin{figure}[h]
\centering
\ovalbox{
\begin{tabular}{c|c|cc}
a)
 \begin{tabular}{ccc}
$(\l x.t)u$&$\tom$&$t[x/u]$\\\\
\begin{tikzpicture}[vcenter]
\node [eport] (out){};
\node at (out.center) [mport, below left=\stalt and \stlar] (fun){};
\node at (out.center) [eport, below right=\stalt and \stlar] (arg){};
\ltensv{out}{fun}{arg}{ap};
\inetcell[inductiveTr,at=(arg.center), below=\sepboxshort](argbody){\scriptsize $\lt{u}$}[180]    

\node at (fun.center) [below right=\stalt and \hstlar, eport](downr){};

\inetcell[inductiveTr,at=(downr.center), below=\sepboxshort](body){\scriptsize $\lt{t}$}[180]    
\node at (body.left pax) [eport, below=\sepboxshort](x){};
\node at (body.middle pax) [eport, below=\sepbox](freevarmiddle){};

\lparvanglefake{fun}{x}{downr}{par}{out=180, in =225, looseness=1}
\aboxvnodes{par}{freevarmiddle}{exbox}{15pt}{25pt}

\node at (body.middle pax) [eport, below=\sepbox](freevarmiddle){};

\lparvangle{fun}{x}{downr}{par}{out=180, in =225, looseness=1}
\draw[exppure] (body.middle pax) to (freevarmiddle);

\node at (argbody.right pax)[eport, below=\sepboxshort](argfreevar){};
\node at (argbody.middle pax)[eport, below=\sepboxshort](argfreevar2){};

\node at (body.right pax) [eport, below right =3*\sepbox and 2*\sepbox](freevarright1){};
\node at (freevarright1.center) [eport, right =\hstlar](freevarright2){};
\gdots{freevarright1}{freevarright2}{}

\node at (body.right pax) [nospace, left =\sepbox/3](freevarright1dummy){};
\node at (body.right pax) [nospace, right =\sepbox/3](freevarright2dummy){};

\draw[exppure,enaryedge, in=135, out=-60] (freevarright1dummy) to (freevarright1);
\draw[exppure,enaryedge, in=135, out=-60] (freevarright2dummy) to (freevarright2);

\node at (argbody.left pax) [nospace, left =\sepbox/3](argfreevarright1dummy){};
\node at (argbody.left pax) [nospace, right =\sepbox/3](argfreevarright2dummy){};

\draw[exppure,enaryedge, in=45, out=-90] (argfreevarright1dummy) to (freevarright1);
\draw[exppure,enaryedge, in=45, out=-90] (argfreevarright2dummy) to (freevarright2);

\node at (freevarright1freevarright2dots.center)[ below=\sepbox, rotate=90, nospace](bracket){\LARGE \{};
\node at (bracket.center)[ below=1.3*\sepbox, nospace](bracket2){\scriptsize $\fv{t}{\cap}\fv{u}$};

\end{tikzpicture}
&
$\tom$
&

\begin{tikzpicture}[vcenter]
\node at (0,0)[eport](out){};
\inetcell[inductiveTr,at=(out.center),below=\sepboxshort](argbody){\scriptsize $\lt{t}$}[180]    

\node at (argbody.middle pax)[eport, below=\sepboxshort](argfreevar2){};
\node at (argbody.left pax)[eport, below=\sepboxshort](vargluarg){};

\inetcell[inductiveTr,at=(argfreevar2.center), below=\sepboxshort](body){\scriptsize $\lt{u}$}[180]    
\node at (body.middle pax)[eport, below=\sepboxshort](bodympax){};
\node at (body.left pax)[eport, below=\sepboxshort](bodylpax){};

\node at (body.right pax) [eport, below right =3*\sepbox and 2*\sepbox](freevarright1){};
\node at (freevarright1.center) [eport, right =\hstlar](freevarright2){};
\gdots{freevarright1}{freevarright2}{}

\node at (body.right pax) [nospace, left =\sepbox/3](freevarright1dummy){};
\node at (body.right pax) [nospace, right =\sepbox/3](freevarright2dummy){};

\draw[exppure,enaryedge, in=135, out=-60] (freevarright1dummy) to (freevarright1);
\draw[exppure,enaryedge, in=135, out=-60] (freevarright2dummy) to (freevarright2);

\node at (argbody.right pax) [nospace, left =\sepbox/3](argfreevarright1dummy){};
\node at (argbody.right pax) [nospace, right =\sepbox/3](argfreevarright2dummy){};

\draw[exppure,enaryedge, in=45, out=-90] (argfreevarright1dummy) to (freevarright1);
\draw[exppure,enaryedge, in=45, out=-90] (argfreevarright2dummy) to (freevarright2);

\node at (freevarright1freevarright2dots.center)[ below=\sepbox, rotate=90, nospace](bracket){\LARGE \{};
\node at (bracket.center)[ below=1.3*\sepbox, nospace](bracket2){\scriptsize $\fv{t}{\cap}\fv{u}$};

\end{tikzpicture}
\end{tabular}
&
 b)
\begin{tikzpicture}[ocenter]

 \node at (0,0) [eport](weaknodeqsd){};
\lweakv{weaknodeqsd}{weakqsd};

 \node at (weaknodeqsd.center) [below left= \stalt and \stlar, mport](tensorpalqsd){};
\node at (weaknodeqsd.center) [below right = \stalt and \stlar, eport](tensorauxqsd){};
\ltensv{weaknodeqsd}{tensorpalqsd}{tensorauxqsd}{tensorqsd}

 \node at (tensorauxqsd.center) [below= \stalt, mport](bangpalqsd){};
\lbangv{bangpalqsd}{tensorauxqsd}{bangqsd}

 \node at (bangpalqsd.center) [below= \stalt, eport](derpalqsd){};
\lderv{derpalqsd}{bangpalqsd}{derqsd}

\node at (tensorpalqsd.center) [below right = 1.5*\stalt and \hstlar, eport](parpaxqsd){};

\node at (weaknodeqsd.center) [right=3*\stlar,eport](weaknodeqsd2){};
\lweakv{weaknodeqsd2}{weakqsd2};

 \node at (weaknodeqsd2.center) [below left= \stalt and \stlar, mport](tensorpalqsd2){};
\node at (weaknodeqsd2.center) [below right = \stalt and \stlar, eport](tensorauxqsd2){};
\ltensv{weaknodeqsd2}{tensorpalqsd2}{tensorauxqsd2}{tensorqsd2}

 \node at (tensorauxqsd2.center) [below= \stalt, mport](bangpalqsd2){};
\lbangv{bangpalqsd2}{tensorauxqsd2}{bangqsd2}

 \node at (bangpalqsd2.center) [below= \stalt, eport](derpalqsd2){};
\lderv{derpalqsd2}{bangpalqsd2}{derqsd2}

\node at (tensorpalqsd2.center) [below right = 1.5*\stalt and \hstlar, eport](parpaxqsd2){};

\lparvangle{tensorpalqsd}{parpaxqsd2}{parpaxqsd}{parqsd}{in=-120,out=-90, looseness=2.5, overlay}
\lparvangle{tensorpalqsd2}{parpaxqsd}{parpaxqsd2}{parqsd2}{in=-90,out=-90, looseness=2.5, overlay}

\node at (weaknodeqsd.center) [left=2*\stlar,eport](out){};
 \node at (out.center) [below= \stalt, mport](bangpalqsd3){};
\lbangv{bangpalqsd3}{out}{bangqsd3}

 \node at (bangpalqsd3.center) [below= \stalt, eport](derpalqsd3){};
\lderv{derpalqsd3}{bangpalqsd3}{derqsd3}

\node at (derpalqsd2.center) [below=33pt, nospace](border){};
\end{tikzpicture}

&
c)
\begin{tikzpicture}[ocenter]

 \node at (0,0) [eport](weaknodeqsd){};

 \node at (weaknodeqsd.center) [below left= \stalt and \stlar, mport](tensorpalqsd){};
\node at (weaknodeqsd.center) [below right = \stalt and \stlar, eport](tensorauxqsd){};
\ltensv{weaknodeqsd}{tensorpalqsd}{tensorauxqsd}{tensorqsd}

\node at (tensorpalqsd.center) [below right = 1.5*\stalt and \hstlar, eport](parpaxqsd){};

\node at (parpaxqsd.center) [below= \stalt, mport](bangpalqsd){};
\lbangv{bangpalqsd}{parpaxqsd}{bangqsd}

 \node at (bangpalqsd.center) [below= \stalt, eport](derpalqsd){};
\lderv{derpalqsd}{bangpalqsd}{derqsd}

\lparvangle{tensorpalqsd}{tensorauxqsd}{parpaxqsd}{parqsd}{in=-120,out=-90, looseness=2.5, overlay}

\node at (parqsd) [left= \stlar/3, dummynode](spacingNode){};

\end{tikzpicture}

\end{tabular}}
\caption{a) A $\tom$-step on terms and on nets. b-c) Counter-examples to correctness without $\parr$-boxes\label{fig:counter}}
\end{figure}
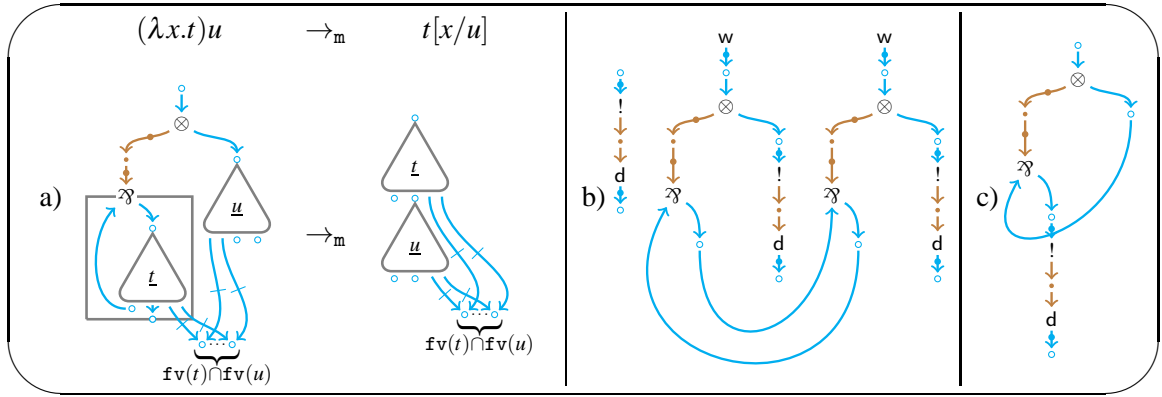

Strong bisimulations preserve reduction lengths, so they preserve divergent/normalizing reductions, and termination properties in general. 

\emph{Technical digression about confluence}. For confluence the point is slightly more delicate, since in general it is preserved only modulo the quotient induced by the strong bisimulation. But mild additional hypothesis allow to transfer confluence. Given two rewriting systems $(S_1, \to)$ and $(S_2, \leadsto)$ and a strong bisimulation $\equiv$ (defined on all terms of $S_1$ and $S_2$),  to transfer confluence from $S_1$ to $S_2$ it is enough to ask that if $s_1\equiv s_2$ and $s_1\to s_1'$ then there is a unique $s_2'$ s.t. $s_2\leadsto s_2'$ and $s_2\equiv s_2'$, see \cite{phdaccattoli} (pp. 83-86) for more details. It is easily seen that in our case the translation enjoys this property in both directions.

These observations (and confluence of $\lvker$) prove:

\begin{corollary}
Let $t\in\lvker$ and $X$ a set of variables. Then $t$ is weakly normalizing/strongly normalizing/a normal form/without a normal form iff $\lamtonetsvar{t}{X}$ is. Moreover, proof nets are confluent.
\end{corollary}

Actually, the translation is more than a strong bisimulation: the reduction graphs\footnote{\emph{Reduction graphs}, which are the graphs obtained considering all reductions starting from a given object, \textit{are not nets}.} of $t$ and $\lamtonets{t}$ are \textit{isomorphic}, not just strongly bisimilar. An easy but tedious refinement of the proof of Theorem \ref{tm:str-bis} proves:

\begin{theorem}[dynamic isomorphism]
Let $t$ be a term and $X$ a set of variables containing $\fv{t}$. The translation induces a bijection $\phi$ between the redexes of $t$ and the redexes of $\lamtonetsvar{t}{X}$ s.t. $R: t\Rew{a} t'$ if and only if $\phi(R): \lamtonetsvar{t}{X}\Rew{a} \lamtonetsvar{t'}{X}$, where $a\in\set{\msym,\esym}$.
\end{theorem}

A nice by-product of the strong bisimulation approach is that preservation of correctness by reduction \textit{comes for free}, since any reduct of a proof-net is the translation of a term.

\begin{corollary}[preservation of correctness]
Let $G$ be a proof net and $G\Rew{} G'$. Then $G'$ is correct.
\end{corollary}

\paragraph{The original boring translation.} For the sake of completeness, Figure \ref{fig:ord-trans} sketches the ordinary CBV translation from $\l$-terms (possibly with iterated applications) to proof nets (including the case for explicit substitutions and using a traditional syntax with boxes on $!$).  An easy computation shows that the term $t=\delta (yz) \delta$, where $\delta=\l x. xx$ maps to a net without normal form, while $t$ is a $\lv$-normal form (see \cite{AccLinearity} for more details). This mismatch is the motivation behind our work.
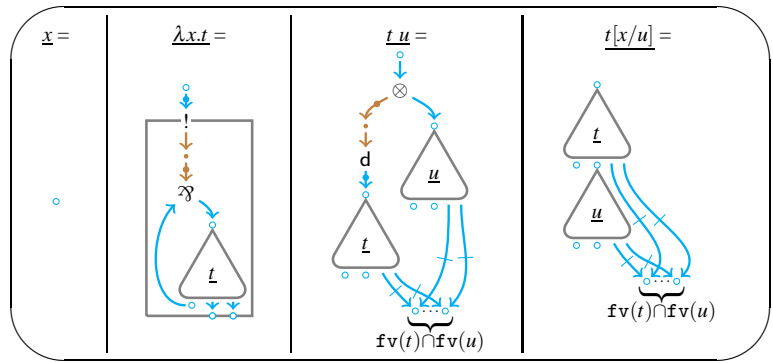
\begin{figure}[t]
\centering
\ovalbox{
\begin{tabular}{c\spc|\spc c\spc|\spc c\spc|\spc c\spc}
\scriptsize$\lamtonets{x}$ = &
\scriptsize $\lamtonets{\l x.t}$  = &
\scriptsize $\lamtonets{t\ u}$ = &
\scriptsize $\lamtonets{t[x/u]}$ = 
\\

\begin{tikzpicture}[ocenter]
\node at (0,0) [eport](up){};

\end{tikzpicture}
&
\begin{tikzpicture}[ocenter]
\node at (0,0) [eport](out){};
\node at (out.center) [below=\stalt, mport](up){};
\lbangvfake{up}{out}{bang}
\node at (up.center) [below right=\stalt and \hstlar, eport](downr){};

\inetcell[inductiveTr,at=(downr.center), below=\sepboxshort](body){\scriptsize $\lamtonets{t}$}[180]    
\node at (body.left pax) [eport, below=\sepboxshort](x){};
\node at (body.middle pax) [eport, below=\sepbox](freevarmiddle){};
\node at (body.right pax) [eport, below=\sepbox](freevarright){};

\lparvanglefake{up}{x}{downr}{par}{out=180, in =225, looseness=1}

\aboxvnodes{bang}{freevarright}{exbox}{15pt}{25pt}
\lbangv{up}{out}{bang}

\node at (body.middle pax) [eport, below=\sepbox](freevarmiddle){};
\node at (body.right pax) [eport, below=\sepbox](freevarright){};
\lparvangle{up}{x}{downr}{par}{out=180, in =225, looseness=1}
\draw[exppure] (body.middle pax) to (freevarmiddle);
\draw[exppure] (body.right pax) to (freevarright);

\end{tikzpicture}
&

\begin{tikzpicture}[vcenter]
\node [eport] (out){};
\node at (out.center) [mport, below left=\stalt and 0.6*\stlar] (fun){};
\node at (out.center) [eport, below right=\stalt and 0.6*\stlar] (arg){};
\ltensv{out}{fun}{arg}{ap};
\inetcell[inductiveTr,at=(arg.center), below=\sepboxshort](argbody){\scriptsize $\lamtonets{u}$}[180]    
\node at (fun.center) [below=\stalt, eport](down2){};
\lderv{down2}{fun}{der}

\node at (argbody.middle pax)[eport, below=\sepboxshort](argfreevar2){};

\node at (argbody.left pax)[eport, below=\sepboxshort](vargluarg){};

\inetcell[inductiveTr,at=(down2.center), below=\sepboxshort](body){\scriptsize $\lamtonets{t}$}[180]    
\node at (body.middle pax)[eport, below=\sepboxshort](bodympax){};
\node at (body.left pax)[eport, below=\sepboxshort](bodylpax){};

\node at (body.right pax) [eport, below right =3*\sepbox and 2*\sepbox](freevarright1){};
\node at (freevarright1.center) [eport, right =\hstlar](freevarright2){};
\gdots{freevarright1}{freevarright2}{}

\node at (body.right pax) [nospace, left =\sepbox/3](freevarright1dummy){};
\node at (body.right pax) [nospace, right =\sepbox/3](freevarright2dummy){};

\draw[exppure,enaryedge, in=135, out=-60] (freevarright1dummy) to (freevarright1);
\draw[exppure,enaryedge, in=135, out=-60] (freevarright2dummy) to (freevarright2);

\node at (argbody.right pax) [nospace, left =\sepbox/3](argfreevarright1dummy){};
\node at (argbody.right pax) [nospace, right =\sepbox/3](argfreevarright2dummy){};

\draw[exppure,enaryedge, in=45, out=-90] (argfreevarright1dummy) to (freevarright1);
\draw[exppure,enaryedge, in=45, out=-90] (argfreevarright2dummy) to (freevarright2);

\node at (freevarright1freevarright2dots.center)[ below=\sepbox, rotate=90, nospace](bracket){\LARGE \{};
\node at (bracket.center)[ below=1.3*\sepbox, nospace](bracket2){\scriptsize $\fv{t}{\cap}\fv{u}$};

\end{tikzpicture}

&
\begin{tikzpicture}[vcenter]
\node at (0,0)[eport](out){};
\inetcell[inductiveTr,at=(out.center),below=\sepboxshort](argbody){\scriptsize $\lamtonets{t}$}[180]    

\node at (argbody.middle pax)[eport, below=\sepboxshort](argfreevar2){};
\node at (argbody.left pax)[eport, below=\sepboxshort](vargluarg){};

\inetcell[inductiveTr,at=(argfreevar2.center), below=\sepboxshort](body){\scriptsize $\lamtonets{u}$}[180]    
\node at (body.middle pax)[eport, below=\sepboxshort](bodympax){};
\node at (body.left pax)[eport, below=\sepboxshort](bodylpax){};

\node at (body.right pax) [eport, below right =3*\sepbox and 2*\sepbox](freevarright1){};
\node at (freevarright1.center) [eport, right =\hstlar](freevarright2){};
\gdots{freevarright1}{freevarright2}{}

\node at (body.right pax) [nospace, left =\sepbox/3](freevarright1dummy){};
\node at (body.right pax) [nospace, right =\sepbox/3](freevarright2dummy){};

\draw[exppure,enaryedge, in=135, out=-60, overlay] (freevarright1dummy) to (freevarright1);
\draw[exppure,enaryedge, in=135, out=-60, overlay] (freevarright2dummy) to (freevarright2);

\node at (argbody.right pax) [nospace, left =\sepbox/3](argfreevarright1dummy){};
\node at (argbody.right pax) [nospace, right =\sepbox/3](argfreevarright2dummy){};

\draw[exppure,enaryedge, in=45, out=-90, overlay] (argfreevarright1dummy) to (freevarright1);
\draw[exppure,enaryedge, in=45, out=-90, overlay] (argfreevarright2dummy) to (freevarright2);

\node at (freevarright1freevarright2dots.center)[ below=\sepbox, rotate=90, nospace](bracket){\LARGE \{};
\node at (bracket.center)[ below=1.3*\sepbox, nospace](bracket2){\scriptsize $\fv{t}{\cap}\fv{u}$};

\end{tikzpicture}
\end{tabular}

}
\caption{\label{fig:ord-trans} the ordinary CBV translation from terms to nets.}
\end{figure}

\section{Motivating \mathintitle{$\parr$}{par}-boxes}
\label{s:par}
The two encodings of $\l$-calculus can be seen as fragments of Intuitionistic Multiplicative and Exponential Linear Logic (IMELL). Let us stress that in IMELL what we noted $\otimes$ and $\parr$ correspond to the right and left rules for the linear implication $\multimap$, and not to the left and right rules for $\otimes$ (the four rules for $\otimes$ and $\multimap$ are collapsed in LL but not in Intuitionistic LL, in particular our $\parr$ acts on the output of the term, \ie\ on the right of the sequent, and corresponds to the right rule for $\multimap$). 

Our argument is that in IMELL there is no correctness criterion unless the syntax is extended with boxes for both $!$ \emph{and} $\multimap$ (our $\parr$), as we shall explain in the next paragraphs. The fragment of IMELL encoding the CBN $\l$-calculus is a special case where the box for $\multimap$ needs not to be represented. The fragment encoding the CBV $\l$-calculus is a special case where the box for $!$ needs not to be represented. So, the two encodings are dual with respect to the use of boxes, and then there is nothing exotic in our use of $\parr$-boxes.

The difficulty of designing a correctness criterion for IMELL is given by the presence of weakenings, which break connectedness. In most cases weakenings simply prevent the possibility of a correctness criterion. The fragment encoding the CBN $\l$-calculus, and more generally Polarized Linear Logic, are notable exceptions. For the encoding of the CBN $\l$-calculus there exist two correctness criteria. Let us show that none of them works for the CBV $\l$-calculus.

The first is the Danos-Regnier criterion, in the variant replacing connectedness with the requirement that the number of connected components of every switching graph is $1+\#w$, where $\#w$ is the number of weakenings at level 0 (after the collapse of $!$-boxes) \cite{Reg:Thesis:92}. In our case this criterion does not work: the net in Fig. \ref{fig:counter}.b verifies the requirement while it does not represent any proof or term.
The second criterion is Olivier Laurent's polarized criterion, because the CBN encoding is polarized. In its original formulation it cannot be applied to the encoding of the CBV $\l$-calculus, because such a fragment is not polarized (there can be a weakening as a premise of a tensor, which is forbidden in polarized logic).
Our re-formulation of Laurent's criterion rejects the net in Figure \ref{fig:counter}.b (because the two $\parr$-links form a cycle), but without using $\parr$-boxes it would accept the net in Figure \ref{fig:counter}.c, which is not correct\footnote{The net in Figure \ref{fig:counter}.c would be rejected by the original version of the criterion, which is based on a different orientation. But the original orientation cannot be applied to our fragment.}.

Thus, the known criteria do not work and there is no criteria for IMELL. The usual way to circumvent problems about correctness is to add some information to the graphical representation, under the form of boxes (as we did) or jumps (\ie\ additional connections). It is well known that in these cases various criteria can be used, but this extra information either is not canonical or limits the degree of parallelism. Another possible solution is to modify the logical system adding the mix rules. However, such rules are debatable, and also give rise to a bad notion of subnet (for details see \cite{phdaccattoli}, pp. 199-201).

Let us stress that our counter-examples to the known criteria do not rely on the exponentials (\ie\ non-linearity): it is easy to reformulate them in Intuitionistic Multiplicative Linear Logic (IMLL) with units\footnote{Just replace each sequence of a ! over a dereliction with an axiom, and the weakenings with $\bot$-links.}, for which then there is no correctness criterion. 

In the case studied in this paper the use of $\parr$-boxes does not affect the level of parallelism in a sensible way. Indeed, in IMELL the parallelism given by proof nets concerns the left rules (of $\otimes$ and $\multimap$, plus contractions and weakenings) and cuts: in our case there is no $\otimes$ (remember our $\otimes$ and $\parr$ rather correspond to the rules for $\multimap$), our technical choices for variables keep the parallelism for contraction and weakenings, and the parallelism of the left rule for $\multimap$ (our $\otimes$) and cuts is preserved (it is given by the equations in (\ref{eq:quotient}), page \pageref{eq:quotient}).
\section{Proof nets: the literature on term representations}
\label{s:history}
When relating $\l$-terms and proof nets a number of technical choices are possible:
\begin{enumerate}
\item \emph{Explicit substitutions}: proof nets implement a $\beta$-step by two cut-elimination steps. This refined evaluation can be seen on the calculus only if the syntax is extended with explicit substitutions.
\item \emph{Variables}: to properly represent variables it is necessary to work modulo associativity and commutativity of contractions, neutrality of weakening with respect to contraction, and permutations of weakenings and contractions with box-borders. In the literature there are two approaches: to explicitly state all these additional congruences or to use a syntax naturally quotienting with respect to them. Such a syntax uses n-ary $?$-links collapsing weakening, dereliction and contractions and delocalizing them out of boxes. It is sometimes called \emph{nouvelle syntaxe}.
\item \emph{Axioms}: various complications arise if proof nets are presented with  explicit axiom and cut links. They can be avoided by working modulo cuts on axioms, which is usually done by employing an interaction nets presentation of proof nets.
\item \emph{Exponential cut-elimination}: the cut-elimination rules for the exponentials admit many presentations. Essentially, either they are big-step, \ie\ an exponential cut is eliminated in one shot (making many copies of the $!$-premise of the cut), or they are small-step, with a rule for each possible $?$-premise (weakening, dereliction, contraction, axiom, box auxiliary port).
\end{enumerate}

We now list the works in the literature which are closer in spirit to ours, \ie\ focusing on the representation of $\l$-calculi into proof nets (and for space reasons we omit many other interesting works, as for instance \cite{DBLP:conf/ifl/Mackie05}, which studies the representation of \emph{strategies}, not of \emph{calculi}). The first such works were the Ph.D. thesis of Vincent Danos \cite{Danos:Thesis:90} and Laurent Regnier \cite{Reg:Thesis:92}, which focused on the call-by-name (CBN) translation. Danos and Regnier avoid explicit substitutions, use n-ary contractions, explicit axioms, and big-step exponential rules, see also \cite{Danos95proof-netsand}. They  characterize the image of the translation using the variant on the Danos-Regnier correcteness criterion which requires that any switching graph has $\# w+1$ connected components, where $\# w$ is the number of weakenings. In \cite{DBLP:journals/tcs/DanosR99} Danos and Regnier use the CBV translation\footnote{Let us point out that \cite{DBLP:journals/tcs/DanosR99} presents an oddity that we believe deserves to be clarified. The authors show that an optimized geometry of interaction for the proof nets of the CBV-translation is isomorphic to Krivine'
s abstract machine (KAM): this is quite puzzling, because 
the KAM is CBN, while they use the CBV translation.}. Both translations are injective.

In \cite{DBLP:journals/tcs/Laurent03,phdlaurent} Olivier Laurent extends the CBN translation to represent (the CBN) $\l\mu$-calculus. He does not use explicit substitutions nor n-ary $?$-links, while he employs explicit axiom links and small-step exponential rules. His work presents two peculiar  points. First, the translation of $\l\mu$-terms is not injective, because---depending on the term---the $\mu$-construct may have no counterpart on proof nets. This induces some mismatches at the dynamic level. Second, Laurent finds a simpler criterion, exploiting the fact that the fragment encoding (the CBN) $\l\mu$-calculus is polarized. In \cite{phdlaurent} Laurent also show how to represent the CBV $\l\mu$-calculus. However, such a representation does not use the same types of the boring translation, as $A\to B$ maps to $?!(A\multimap B)$, and not to $!(A\multimap B)$.

Lionel Vaux \cite{phdvaux} and Paolo Tranquilli \cite{tranquillithesis,DBLP:journals/tcs/Tranquilli11} study the relationship between the differential $\l$-calculus and differential proof nets. Vaux also extends the relationship to the classical case (thus encompassing a differential $\l\mu$-calculus), while Tranquilli refines the differential calculus into a \emph{resurce calculus} which better matches  proof nets. They do not use explicit substitutions, nor n-ary contractions, while they use interaction nets (so no explicit axioms and cut link) and small-step exponential rules. Both Tranquilli and Vaux rely on the Danos-Regnier criterion, despite the fragment encoding their calculi is polarized and can be captured using Laurent's criterion by using boxes for coderelictions; in the context of $\l$-calculus such boxes do not reduce the parallelism of the representation.

Delia Kesner and co-authors \cite{DBLP:conf/lics/CosmoK97,DBLP:journals/mscs/CosmoKP03,DBLP:journals/iandc/KesnerL07} study the relationship with explicit substitutions (in the CBN case). The main idea here is that explicit substitutions correspond to exponential cuts. They use explicit axiom links and small-step exponential rules, but they do not employ n-ary contractions (and so they need additional rules and congruences). Because of explicit substitutions the translation is not injective: now different terms may map to the same proof net, as in this paper. They do not deal with correctness.

In none of these works the translation is a strong bisimulation. In \cite{DBLP:conf/csl/AccattoliG09} the author and Stefano Guerrini use a syntax inspired by proof nets (and extended with jumps) to represent the CBN $\l$-calculus with explicit substitutions. That work is the only one employing (the equivalent of) n-ary $?$-links and (the equivalent of) small-step exponential rules. In \cite{DBLP:conf/csl/AccattoliG09} the correctness criterion is a variation over Lamarche's criterion for essential nets, which relies in an essential way on the use of jumps. A reformulation in the syntactic style of this paper of both \cite{DBLP:conf/csl/AccattoliG09} and of Danos and Regnier's proof nets for the CBN $\l$-calculus can be found in \cite{phdaccattoli}, together with a detailed account of the strong bisimulation.

Here, hypergraphs allow us to use n-ary $?$-links and collapse axioms and cut links (as if we were using interaction nets). More precisely, we represent n-ary $?$-links by allowing $e$-nodes to have more than one incoming link. This choice overcomes some technicalities about \emph{gluing} and \emph{de-gluing} of $?$-links. Such technicalities are always omitted, but they are in fact necessary to properly define subnets and cut-elimination. We also employ big-step exponential rules and explicit substitutions.
\ms

\textbf{Acknowledgements.}
To Stefano Guerrini, for introducing me to proof nets, correctness and the representation of $\l$-terms, 
and to Delia Kesner, for helping with the financial support of this work.

{

}
\end{document}